\documentclass[10pt]{amsart}
\usepackage{enumerate}
\usepackage{color}
\usepackage{dsfont}
\usepackage{graphicx}
\usepackage{hyperref}
\usepackage{caption}
\calclayout

\newtheorem{theorem}{Theorem}[section]

\newtheorem{lemma}[theorem]{Lemma}

\newtheorem{definition}[theorem]{Definition}
\newtheorem{remark}[theorem]{Remark}

\newtheorem{assumption}[theorem]{Assumption}

 \usepackage{color}
\usepackage{natbib}

\newcommand{\R}{\mathbb{R}}
\newcommand{\Q}{\mathbb{Q}}

\newcommand{\N}{\mathbb{N}}

\newcommand{\E}{\mathbb{E}}

\renewcommand{\P}{\mathds{P}}

\newcommand{\coloneq}{\mathrel{\mathop:}=}

\renewcommand{\epsilon}{\varepsilon}
\newcommand{\opt}{{\star}}

\newcommand{\interior}[1]{{#1}^{o}}
\newcommand{\cA}{\mathcal{A}}
\newcommand{\cF}{\mathcal{F}}
\newcommand{\cP}{\mathcal{P}}
\newcommand{\cS}{\mathcal{S}}
\newcommand{\cD}{\mathcal{D}}
\renewcommand{\d}{p_d}
\newcommand{\hatd}{\hat{p}_d}

\newcommand{\Adm}{\mathcal{A}}

\newcommand{\Var}{\mathrm{Var}}
\newcommand{\pip}{(\pi^\opt)^\prime}

\begin{document}

\title[Robust marginal utility pricing]{Distributionally robust portfolio maximisation and marginal utility pricing in one period financial markets}

\date{\today}

\author[Jan Ob{\l}{\'o}j]{Jan Ob{\l}{\'o}j$^*$}
\address{Jan Ob{\l}{\'o}j\newline
Mathematical Institute and St John's College, University of Oxford\newline
Woodstock Road\newline
Oxford, OX2 6GG
}
\author[Johannes Wiesel]{Johannes Wiesel$^{*+}$}
\address{Johannes Wiesel\newline
Department of Statistics, Columbia University\newline
1255 Amsterdam Avenue\newline
New York, NY 10027
}

\thanks{$^*$The authors would like to thank Daniel Bartl and Samuel Drapeau for helpful discussions. The authors declare that there are no conflicts of interests. JW acknowledges support from the German Academic Scholarship Foundation.\\
$^{+}$Corresponding author: Johannes Wiesel, Department of Statistics, Columbia University, 1255 Amsterdam Avenue, New York, NY 10027}

\begin{abstract}
We consider the optimal investment and marginal utility pricing problem of a risk averse agent and quantify their exposure to a small amount of model uncertainty. Specifically, we compute explicitly the first-order sensitivity of their value function, optimal investment policy and Davis' option prices to model uncertainty. To achieve this, we capture model uncertainty by replacing the baseline model $\P$ with an adverse choice from a small Wasserstein ball around $\P$ in the space of probability measures. Our sensitivities are thus fully non-parametric. We show that the results entangle the baseline model specification and the agent's risk attitudes. The sensitivities can behave in a non-monotone way as a function of the baseline model's Sharpe's ratio, the relative weighting of assets in an agent's portfolio can change and marginal prices can increase when an agent faces model uncertainty. 
\end{abstract}

\maketitle

\emph{This paper is dedicated to the memory of Mark H.A. Davis. I remain forever grateful to Mark who introduced me to the world of mathematical finance when I joined his group at Imperial College London in 2006. He was a mentor and a friend. His intellectual curiosity combined with a generosity of spirit and good humour were hugely enriching. His capacity to ask key questions and find elegant and insightful answers set my gold standard.\\
\begin{flushright}
 Jan Ob\l\'oj 
\end{flushright}
}
\section{Introduction}

Modern theory of finance enables powerful and versatile quantitative analysis of optimal investment problems. It would be impossible to give justice here to the body of relevant literature. However, since the seminal works of \cite{markowitz1959portfolio, merton1969lifetime}, at its heart are mathematical techniques which compute an agent's investment decisions from two inputs: the agent's risk attitudes and a probabilistic description of the financial market. The former is considered subjective and an important stream of research looks at ways to elucidate an agents' preferences and whether these can be assumed to be \emph{rational}, i.e., to satisfy certain axiomatic properties. The latter input, a probability measure $\P$, captures the agent's model for the financial market and results from numerous considerations and tradeoffs, such as calibration to empirical data, reproduction of certain stylised features and analytical tractability, among others.

With a fixed baseline model $\P$, a rational decision maker, following \cite{merton1969lifetime} and in line with the classical decision theory, see e.g., \cite{von1947theory, savage1951theory}, finds her optimal investment policy $\pi^\opt$ by maximising her expected utility of wealth
$$\max_{\pi\in \cA} \E_{\P}[u(X^\pi)],$$ 
and computes the prices of derivative instruments through marginal utility pricing, see \cite{davis1997option}. However, in practice, the agent is bound to have a degree of uncertainty about their choice of $\P$. This is often referred to as \emph{Knightian uncertainty} after \cite{knight2012risk}. Following \cite{anderson2003quartet}, we are interested here in misspecifications of $\P$ which are small, in a statistical sense, but unrestricted otherwise (e.g., these do not need to be absolutely continuous with respect to $\P$). Our main contribution is to understand and quantify the sensitivity of an agent's decisions to such model uncertainty. We compute explicitly the first order change to  the key quantities -- the agent's value function, the agent's optimal portfolio allocation and their marginal utility prices -- in response to small levels of model uncertainty.  

From an axiomatic point of view, see \cite{gilboa1989maxmin, maccheroni2006ambiguity,schied2007optimal, schied2009robust}, an agent who is averse both to risk and model uncertainty, considers a max-min criterion
$$ \max_{\pi\in \cA} \min_{\tilde \P\in \cP}\E_{\tilde \P}[u(X^\pi)],$$
see section \ref{sec:uncertainty_ways} for discussion and generalisations. While appealing from a decision theoretic point of view, this criterion only introduces an additional subjective input: the family $\cP$ of plausible models for reality. Furthermore, it also often renders the problem intractable with little hope for analytic formulae. 
In this paper we offer a tried and tested mathematical solution to this problem: we compute the first-order approximation to the mapping giving outputs in function of $\cP$, i.e., we compute their derivative with respect to $\cP$ at the point $\cP=\{\P\}$. Specifically, as $\cP$ varies in an infinite dimensional space, we consider $\cP$ given as a ball around $\{\P\}$ with radius $\delta$ and compute explicitly the first order expansion of all of the outputs in $\delta$. This, we believe, captures the essence of model misspecification as described above, see \cite{anderson2003quartet}. 

We use the Wasserstein metric to capture balls around $\{\P\}$. We argue that this metric is the natural choice both because of its theoretical properties as well as from a data-driven perspective, see section \ref{sec:main results}. Our analysis is fully non-parametric and reveals that the sensitivity to model uncertainty is a complex issue, which entangles the baseline model and the agent's risk attitudes. We provide explicit formulae to compute these sensitivities, thus enabling an agent to quantify their exposure to model uncertainty. We believe that these sensitivities offer a more robust and structural tool than any existing, parametric or ad-hoc, methods.

Naturally, the agent's value function is decreasing in $\delta$: the introduction of model uncertainty combined with a max-min approach means the agent will never welcome the new source of uncertainty. However, the agent's optimal investment policy will react in a much more complex way. Already in a one dimensional setting, we show that the reduction in the agent's position is not monotone in the baseline model's Sharpe ratio. In higher dimensions, as the uncertainty affects the individual assets as well as their co-dependence, also the relative weighting of different assets changes. Finally, the agent's marginal utility price of an option, also known as the Davis price after \cite{davis1997option}, can both decrease or increase as model uncertainty appears. It always decreases if an agent was not trading in the baseline model. But if the agent trades, than the option can effectively provide insurance for certain adverse model changes and may become more valuable than before. 

The rest of the paper is organised as follows. We first briefly introduce the baseline setting and its classical analysis. Then, in section \ref{sec:uncertainty_ways}, we give a high-level overview of the literature and topics related to model uncertainty. After that, we proceed to introducing our approach and discuss our main findings. These are then illustrated with worked-out examples in section \ref{sec:examples}. A careful presentation of all the assumptions and theorems then ensues. We look at the optimal investment problem and the marginal utility pricing of options in sections \ref{sec:robust eum} and \ref{sec:robust Davis price} respectively.  Appendix \ref{app:marg_util} presents proofs of the classical results recalled in section \ref{sec:marginal_util}, while the supplement \cite{obloj2021} contains detailed explicit computations for all of the examples of section \ref{sec:examples}.

\section{The baseline optimal investment and option pricing problems}\label{sec:baseline}

We consider a one-period financial market with $d$ risky assets, for some $d\geq 1$. We assume there is also a deterministic riskless asset and, without loss of generality, express all prices in its discounted units. The prices of risky assets today and at time one are given by the vectors $S_0$ and $S_1$ and we let $X=S_1-S_0$.  We fix a closed convex state space $\mathcal{S}\subseteq \R^d$ and we only consider $X\in \cS$. This can be used, e.g., to ensure that in all considered models the prices remain positive. Throughout the paper, \emph{a model} is synonymous to the distribution of $X$: a probability measure supported on $\cS$. We fix one such $\P$ and refer to it as the \emph{baseline model}. To simplify the notations, we work on the canonical probability space $(\Omega, \cF,\P)$ with $\Omega=\cS$, $X(\omega)=\omega$, $\cF_0= \{\emptyset, \Omega\}$, $\cF_1=\cF=\sigma(X)$.
 
An agent is able to trade in the market using strategies $\pi\in \Adm\subseteq \R^d$, where $\Adm$ is closed and convex. 
In particular, if the first $k$ assets are traded with no restrictions and the remaining assets are not traded, then $\Adm=\R^k\times \{0\}^{d-k}$. The scalar product of $X$ and $\pi$ is denoted $\langle X,\pi\rangle=\sum_{i=1}^d \pi_i X_i$. 

\subsection{Optimal investment}\label{sec:EUM}
We consider an agent endowed with a utility function $u$. Throughout the paper we make the following:
\begin{assumption}[Standing Assumption]\label{Ass:4a} \label{Ass:stand}
The state space $\cS$ and the action space $\cA$ are closed convex subsets of $\R^d$. The baseline model $\P$ is supported on $\cS$ and satisfies the following no-arbitrage and non-degeneracy condition:
$$\forall\pi\in \R^d\setminus \{0\}\quad \P (\langle  X,\pi\rangle>0)>0.$$
$\cD$ is an open convex subset of $\R$ and $u:\mathcal{D} \to \R$ is strictly concave, strictly increasing, continuously differentiable and bounded from above. The state, action and domain sets are compatible in that there exists $\epsilon_0>0$ such that  $$\{\langle x, \pi\rangle + \eta :\ x\in \mathcal{S}, \pi\in \Adm, |\eta|<\epsilon_0\}\subseteq \mathcal{D}.$$
We fix a function $g:\cS\to \R$ and assume that either $\cD=\R$ or $g$ is bounded.
\end{assumption}

Consider the optimal investment problem for an agent who is maximising their expected utility of wealth, i.e., 
\begin{equation}\label{eq. Merton}
    V= \sup_{\pi\in \Adm} \E_{\P}\left[u\left(\langle  X,\pi \rangle \right)\right].
\end{equation}
Let us remark that in order to simplify our notation we have not explicitly accounted for interest rates or an initial investment capital in the above formulation. These can be easily incorporated by considering discounted stock prices and composing $u$ with an affine function. Under our assumptions, the optimal portfolio $\pi^\opt$ for $V$ exists, is unique and, for most choices of $\cA$ considered in the literature, belongs to the (relative) interior $\interior{\Adm}$ of $\Adm$. It is characterised through the first order condition
\begin{equation*}
    \E_{\P} \left[ X \cdot u^\prime\left( \langle X, \pi^\opt \rangle \right)\right] =0.
\end{equation*}
We can express this differently by saying that $X$ is a $\Q_u$-martingale, where
\begin{equation}\label{eq:u mart mg}
 \frac{d \Q_u}{d \P} = \frac{u^\prime\left( \langle X, \pi^\opt \rangle\right)}{\E_{\P}\left[u^\prime\left( \langle X, \pi^\opt \rangle\right)\right]}.
\end{equation}
This idea was exploited by \cite{rogers1994equivalent} to obtain a proof of the Dalang-Morton-Willinger theorem in discrete time using a dynamic programming approach. It is also crucial for marginal utility pricing which we discuss next.

\subsection{Marginal utility price}\label{sec:marginal_util}
\cite{davis1997option} proposed an innovative way of option pricing based on marginal utility. Consider an option with payoff $\tilde g(S_0,S_1)$. Note that since $S_0$ is known and $X=S_1-S_0$, we can put $g(x)=\tilde g(S_0,x+S_0)$ and with no loss of generality simply consider an option with payoff $g(X)$. Our standing assumptions on $g$ were included in Assumption \ref{Ass:stand} and further growth assumptions will be added later. We now consider the maximisation problem
\begin{align}\label{eq:def_davis}
V(\epsilon,\d):=\sup_{\pi\in \Adm} \E_{\P}\left[ u\left(-\epsilon+ \langle X, \pi \rangle +\frac{\epsilon}{\d}g(X) \right)\,\right],
\end{align}
where $\d\ge 0, \epsilon\in \R$, and note that $V(0,\d)=V$. 
We note that $\epsilon\mapsto V(\epsilon,\d)$ is well-defined for $\epsilon>0$ small enough since 
\begin{align*}
\left\{ -\epsilon+\langle x, \pi\rangle+\frac{\epsilon}{\d}\, g(x):\ x\in \mathcal{S}, \pi\in \mathcal{A} \right\}\subseteq \mathcal{D}
\end{align*}
for $|\epsilon|$ small enough by Assumption \ref{Ass:stand}.
We denote strategies attaining the supremum in the expression above by $\pi^\opt(\epsilon)$. For notational simplicity and in order to emphasize that $\pi^\opt(0)$ does not depend on $\d$ (nor on $g$),  we still write $\pi^\opt$ for the $V$-optimiser. The marginal utility price is the one which makes the agent indifferent between buying $\epsilon$ options $g$ and the investment problem without the option $g$, where $\epsilon$ is small. This is made precise in the following definition.

\begin{definition}[Marginal utility price of \cite{davis1997option}]\label{def:davis}
\label{def:marginal}
Suppose that for each $\d>  0$, the function $\epsilon \mapsto V(\epsilon,\d)$ is differentiable at $\epsilon=0$ and $\hatd$ is a solution to
\begin{align*}
\partial_{\epsilon} V(0,\d):=\left.\partial_{\varepsilon} V\left(\varepsilon, p_{d}\right)\right|_{\varepsilon=0}=0
\end{align*}
Then $\hatd$ is called a \emph{marginal utility price} of the option $g$.
\end{definition}

Under the assumption on $u$ above and with suitable uniform integrability, $\pi^\opt(\epsilon)$ is unique and $\pi^\opt(\epsilon )\to \pi^\opt(0)$. The key result of \cite{davis1997option} is that $\hatd$ is also unique and satisfies
\begin{align}\label{eq:davis_main}
\hatd=\E_{\Q_u}\left[g(X)\right],
\end{align}
where the precise statements and proofs of these results are given in Appendix \ref{app:marg_util}. 
The Davis price can thus be interpreted as the arbitrage-free model price under the martingale measure $\Q_u$. 
It is a linear pricing rule, a consequence of considering \emph{marginal} pricing. An indifference pricing problem for a given quantity of options would naturally lead to a non-linear pricing rule. We refer to \cite{musiela2008single, henderson2004utility} and the references therein for a detailed discussion. 

While Davis' original article focuses on classical continuous time models (see also \cite{karatzas1996pricing,henderson2004utility, hugonnier2005utility}), Davis pricing in discrete time models has been investigated, e.g., in \cite{schal2000price,schal2000portfolio, schal2002markov, kallsen2002utility, rasonyi2005utility}. 
More broadly, \eqref{eq:davis_main} gives a way to select a martingale measure used for pricing and can be seen as part of a broader literature considering such measures for incomplete markets. 
In the case of exponential utility, $u(x)=-\exp(- \gamma x)$, $\Q_u$ is also the \emph{minimal relative entropy} martingale measure of \cite{frittelli2000minimal, rouge2000pricing}. Indeed, for any $\Q\in \cP(\cS)$, using the explicit density in \eqref{eq:u mart mg}, we have
$$
H(\Q \mid \P) := \E_{\Q}[\ln(d\Q/d\P)] = H(\Q\mid \Q_u) - \E_{\Q}[\gamma\langle X,\pi^\opt\rangle]-\ln \E_{\P}[u^\prime \left( \langle X, \pi^\opt\rangle\right)/\gamma].
$$
Restricting to martingale measures $\Q$, the middle term on the right disappears and $H(\Q\mid \Q_u)\geq 0$ so that
\begin{align*}
\inf_{\Q} H(\Q\mid \P)\geq - \ln \E_{\P}[u^\prime \left( \langle X, \pi^\opt\rangle\right)/\gamma],
\end{align*}
with equality when $\Q=\Q_u$, as required. In the literature on pricing in incomplete markets, other martingale measures have also been discussed, such as the \emph{minimal martingale measure} of \cite{follmer1991hedging} or the \emph{variance optimal measure} of \cite{Schweizer.1996}. 

Despite the abundant literature on this topic, to the best of our knowledge, even in the simple one-period framework we discuss here, the impact of model uncertainty on the Davis price has not been investigated. This is the main motivation behind our work which we now discuss in more detail.

\section{Model misspecification}\label{sec:uncertainty_ways}

Model uncertainty and robustness to model perturbations or misspecification are of paramount importance in any modelling context. In particular, in quantitative finance and economics,  
model uncertainty has been an active topic of research, not least in the wake of the 2008 financial crisis. 
While we can not hope to do justice here to all the relevant works, we will highlight some key developments which put our contributions in their historical perspective.

To guide our discussion, it is helpful to establish, following \cite{hansen2016ambiguity}, a taxonomy of levels of model uncertainty. Going back to \cite{knight2012risk}, but also \cite{keynes1921treatise,arrow1951alternative}, we speak of \emph{risk} as the probabilistic (uncertain) nature of the future states of the world captured \emph{within a model} and we speak of \emph{uncertainty} where this nature is not captured by the proposed model. This could be because we picked a wrong model from a given ensemble: if the question is about which model from a given class to pick we speak of \emph{model ambiguity}. Finally, if we are uncertain about which class of models to use or how to model future outcomes at all, we speak of \emph{model misspecification}. 
The distinction between ambiguity and misspecification may appear artificial. The former is a special case of the latter and, by taking a large infinite dimensional class of models, we can recast the latter as the former, see also \cite[Remark 4.1]{cont2010robustness}. Nevertheless, the distinction often serves as a useful taxonomy to guide a general discussion. 

Without being too prescriptive, we think of model ambiguity when the class of models is restricted to a specific, often parametric, family. The main  advantage of this approach is its tractability. Parameters are typically constants, or processes in a dynamic setting, taking values in some finite dimensional compact set. 
In a seminal contribution \cite{merton1969lifetime} considered the optimal investment problem of an agent trading in risky asset modelled using a geometric Brownian motion. 
Specifically, he considered an agent endowed with a power utility and wanting to maximise the expected utility of their consumption, a criterion justified by the axiomatic decision-theoretic works of \cite{von1947theory,savage1951theory}. \cite{merton1969lifetime} derived the optimal portfolio value and trading strategy explicitly. In consequence, even if Merton worked with a fixed baseline model, his results included \emph{implicitly} a sensitivity analysis with respect to model parameters. Since then, parameter uncertainty has been considered explicitly in a great number of papers, see for example \cite{rogers2001relaxed,chen2002ambiguity,maenhout2004robust,kerkhof2010model,hernandez2006robust,biagini2017robust,balter2020pricing} and the references therein. The classical expected utility maximisation problem, as considered by \cite{merton1969lifetime}, is typically replaced by a maxmin formulation. This paradigm is also widely adopted in the robust control literature, see \cite{sirbu2014note, bayraktar2016robust}, and can be seen as a two-player game setting: the agent picks their best strategy to play against the nature who decides on adverse choice of model $\P$ from a set $\cP$. An axiomatic decision-theoretic justification of this criterion was provided by \cite{gilboa1989maxmin}, building on earlier contributions, including \cite{anscombe1963definition, ellsberg1961risk} and \cite{schmeidler1989subjective}.

We note that in a dynamic multi-period setting, the agent can learn and reduce their model uncertainty. This led \cite{bielecki2019adaptive} to develop an adaptive robust control approach and was exploited in a forward utility setting by \cite{kallblad2018dynamically}. It is also the cornerstone to the Bayesian model averaging and updating approach, see \cite{hoeting1999bayesian} or \cite{karatzas2001bayesian} for a continuous-time approach, where the drift is filtered out. In this approach one fixes a family of models $\mathcal{M}$ whose corresponding parameters are denoted by $\theta \in \Theta(M)$, for $M\in \mathcal{M}$. The Bayesian observer has two levels of prior beliefs, namely a prior distribution on the models denoted by $\mu_{\text{prior}}$ and a prior distribution of the model parameters given the model, denoted by $\P(\theta \mid M)$. The posterior probability distribution $\mu_{\mathrm{post}}$ on the space $\mathcal{M}$ is then computed according to Bayes' rule
\begin{align*}
d\mu_{\mathrm{post}}(M \mid y)=\frac{\int_{\Theta(M)} \P(y\mid \theta , M)d \P(\theta \mid M) \cdot d\mu_{\mathrm{prior}}(M)} {\int_{\mathcal{M}} \int_{\Theta(N)} \P(y\mid \theta , N)d\P(\theta \mid N) d\mu_{\mathrm{prior}}(N)}.
\end{align*}
This gives rise to Bayesian updating of optimal investment problems. While appealing because of its theoretical simplicity, this modelling approach is often too sophisticated from a practical point of view, both to specify the priors and to compute the posteriors, see \cite{cont2006model}.

We classed the works discussed so far as dealing with model ambiguity. We think of model misspecification when possible perturbations of the baseline probability measure are not restricted to a parametric family, although they could be constrained, e.g., via calibration to given  market quoted prices of assets. This typically results in an infinite dimensional set of models to consider. To quote \cite[Section 10]{anderson2003quartet}, ``we envision this approximating model to be analytically tractable, yet to be regarded by the decision maker as not providing a correct model of the evolution of the state vector. The misspecifications we have in mind are small in a statistical sense but can otherwise be quite diverse."
We thus believe that the parametric ambiguity is rarely going to capture the true nature of model uncertainty facing a modeller. When we speak of \emph{model uncertainty} we mean it in the broadest meaning of \emph{model misspecification} and this is our prime interest here. 

We note that \cite{gilboa1989maxmin} did not provide specific insights into how the set of models $\cP$ in the maxmin criterion should be selected. In most of the works mentioned so far, its elements formed a parametric family of probability measures. We start our discussion of the literature dealing with model misspecification using the generic set of models which are absolutely continuous with respect to the baseline measure $\P$: $\cP=\{\tilde{\P}: \tilde{P}\ll \P\}$. It is natural to consider choosing attainable wealth $\langle X, \pi \rangle$ to maximise a Lagrangian criterion
\begin{align}\label{eq:entropy}
\inf_{\tilde{\P}\in \cP} \left[ u(\langle X, \pi \rangle)  +\lambda L\left( \frac{d\tilde{\P}} {d\P} \right)\right]
\end{align}
for some penalty function $L:\R\to \R$. In decision theory, this criterion is known as multiplier preferences, see \cite{maccheroni2006ambiguity}, and includes the maxmin criterion (for a penalty function taking only the values of zero and infinity), see also \cite{hansen2001robust, uppal2003model, lam2016robust}. 
In continuous time models, such an approach can also be used to cover drift uncertainty for SDEs. 
Building on the earlier results of \cite{kramkov1999asymptotic} and the works on risk measures, \cite{schied2005duality} developed this robust optimal investment problem and subsequent works focused on more particular setups, see \cite{hernandez2006robust}. More recently, motivated by the above, 
\cite{cohen2017data}  considers the problem
\begin{align*}
\inf_{\tilde{\P}\in \mathcal{P}} \left( \E_{\tilde{\P}} [u(\langle X, \pi\rangle, \bold{x})]-\alpha_{\mathcal{P}\mid \bold{x}}(\tilde{\P})\right)
\end{align*}
for some observed data $\bold{x}$, a set of measures $\mathcal{P}$ and a likelihood function $\alpha_{\mathcal{P}\mid \bold{x}}$, depending both on the choice of models $\mathcal{P}$ and the data $\bold{x}$. This offers a novel approach to finding a robust estimate for $\E_{\P}[u(\langle X, \pi \rangle)]$, including both the likelihood of a particular model $\tilde{\P}$ and the information coming from the data $\bold{x}$. Finally, note that in the particularly well studied case of $L(x)=\log(x)$, i.e., the relative entropy penalty, an essentially equivalent criterion is given by 
\begin{align*}
\inf_{\tilde{\P}\in B^{\mathrm{KL}}_\delta(\P) } \E_{\tilde{\P}} \left[ u(\langle X, \pi\rangle ) \right],\quad \text{where } B^{\mathrm{KL}}_\delta(\P)=\{\tilde{\P} \ll \P:\ H(\tilde{\P}\mid \P)\le \delta\}
\end{align*}
known as \emph{constraint preferences}, see \cite{hansen2001robust,lam2016robust} and section \ref{sec:WVP} below. While often tractable, approaches of this kind are restrictive since any measure in a ball in $\mathcal{P}$ around $\P$ has to be absolutely continuous with respect to $\P$. 

Recognising the important limitations of the above approaches, a rich stream of literature in mathematical finance focuses on the so-called non-dominated setting for model uncertainty, i.e., considering $\cP$ whose elements are not all absolutely continuous with respect to one baseline model. This started with the so-called uncertain volatility models of 
\cite{avellaneda1995pricing, lyons1995uncertain}. The main focus has since been on pricing and hedging questions from specific option types, e.g., 
\cite{hobson1998robust,cox2011robusta,galichon2014stochastic}, to more holistic duality results, e.g., \cite{beiglbock2013model,bouchard2015arbitrage, hou2018robust}. It is impossible to do this research justice here and we refer the reader to the discussion in \cite{burzoni2019pointwise}. Equally, the maxmin approach to the expected utility maximisation in a non-dominated setting has been considered in a number of works, see  \cite{denis2007utility, rasonyi2005utility, nutz2014utility, neufeld2018robust, carassus2019robust}, usually focusing on existence and uniqueness of optimal investment strategies and dual representations.\\

Finally, we close this short discussion by noting that another stream of literature in which the \emph{de facto used} model is a small perturbation of the baseline model is found among papers on expected utility maximisation under transaction costs, see \cite{cvitanic1996hedging, kallsen2010using, czichowsky2016duality}. Using a convex duality approach these authors derive a solution to the optimal investment problem under proportional transaction cost in continuous time by constructing a so-called shadow price process, i.e., a semi-martingale process taking values within the bid-ask spread region, whose solution to the optimal investment problem in the frictionless market exists and coincides with the solution of the original problem under transaction cost. In this sense the shadow price can be interpreted as a process living in a small neighbourhood of the original price process and the problem with transaction costs is included in a suitable maxmin formulation of the one without transaction costs.

\section{Distributionally robust approach - summary of the main results}
\label{sec:main results}
We propose to capture model misspecification by considering a ball, in the space of probability measures, around the baseline model $\P$. As noted above, this has most often been considered in the economics literature using the Kullblack--Leibler divergence, i.e., the relative entropy, see \cite{lam2016robust} for general sensitivity results, \cite{calafiore2007ambiguous} for applications in portfolio optimisation and \cite{hansen2001robust} and the references therein for a broader context. KL-divergence has good analytic properties and often leads to closed-form solutions. However, it only allows one to consider measures which are absolutely continuous with respect to the baseline model $\P$. In particular, if the latter is the empirical measure supported on $N$ observations, any meaningful analysis requires ad hoc structural assumptions and modifications, creating an additional layer of uncertainty. 

In this paper we propose to use the Wasserstein distance. Coming from the optimal transport theory, this distance is known to lift the natural distance on the state space $\cS$ to the space of measures on $\cS$. It is widely used in many application fields, from image processing, see for example \cite{swoboda2013convex,tartavel2016wasserstein} to statistical analysis of the Wasserstein barycenters, see for example \cite{bigot2018characterization}. We refer to \cite{peyre2019computational} and the references therein for a comprehensive overview of the field. While harder to handle analytically, it is more versatile and does not require any additional structural assumptions. It is well motivated from statistical theory and asymptotic consistency results, see \cite{fournier2015rate, esfahani2015data}. 
 
We first introduce some additional notation. We endow $\R^d$ with the Euclidean norm $|\cdot|$ and write $\interior{\Gamma}$ for the interior of a set $\Gamma$. 
We let $\mathcal{P}(\mathcal{S})$ denote the set of all (Borel) probability measures on $\mathcal{S}$ and, for $p\in [1,\infty)$,  
$$\mathcal{P}_p(\mathcal{S}):=\left\{\P \in \mathcal{P}(\mathcal{S}) \ : \ \E_{\P} [|X|^p] <\infty\right\}.$$
For $\P,\tilde{\P}\in \mathcal{P}_p(\mathcal{S})$, we define the $p$-Wasserstein distance via 
\begin{equation*}
    W_p(\P, \tilde{\P})=\inf\left\{\E_{\gamma}[ |X-Y|^p]\colon \gamma \in \mathrm{Cpl}(\P,\tilde{\P}) \right\}^{1/p},
\end{equation*}
where $\mathrm{Cpl}(\P,\tilde{\P})$ is the set of all probability measures $\gamma$ on $\mathcal{S}\times\mathcal{S}$ with first marginal $\gamma_1:=\gamma(\cdot\times\mathcal{S})=\P$ and second marginal $\gamma_2:=\gamma(\mathcal{S}\times\cdot)=\tilde{\P}$ and $(X,Y)$ is the canonical process on $\cS\times \cS$. We also define the $W_\infty$ distance for $\P,\tilde{P}\in \mathcal{P}(\mathcal{S})$ via
\begin{align}\label{eq:Winfty}
\mathcal{W}^{\infty}(\P, \tilde{\P})&:= \inf_{\gamma\in \mathrm{Cpl}(\P,\tilde{\P})} \gamma\text{-ess-sup }|X-Y|\nonumber \\
&=\inf\left\{\epsilon >0 \ \Big| \ \P(B) \le \tilde{\P}(B^{\epsilon}), \ \tilde{\P}(B) \le \P(B^{\epsilon}) \ \forall B \in \mathcal{B}(\R_+^d) \right\}.
\end{align}
The Wasserstein ball of size $\delta\geq0$ around $\P$ is denoted
\begin{equation*}
    B_{\delta}(\P) = \left\{\tilde{\P} \in \mathcal{P}(\mathcal{S}) :  W_p(\P,\tilde{\P}) \leq \delta  \right\}.
\end{equation*}
Let us fix some $p\in(1,\infty]$ and let $q:=p/(p-1)$ so that $1/p+1/q=1$, with the usual convention that $q=1$ if $p=\infty$. 
\begin{remark}\label{rem:p}
The choice of $p$ is part of the agent's preferences, like the utility function $u$, and has important consequences. As the Wasserstein distances $W_p$ are increasing in $p$, the choice of $p$ influences the size of the Wasserstein ball $B_\delta(\P)$ considered. A more ambiguity-averse agent might choose a lower Wasserstein order than a less ambiguity-averse one. They would then obtain a ball $B_\delta(\P)$ which is larger and thus allows for greater perturbations to the baseline model. On the technical side, we will see in section \ref{sec:robust eum} that the choice of $p$ is linked to the growth rate of the utility function $u$.
\end{remark}

Throughout the rest of the paper we make the following assumption:
\begin{assumption}\label{Ass:4}
$\P\in \cP_p(\cS)$ and the boundary of $\cS\subset \R^d$ has $\P$--zero measure.
\end{assumption}
We consider an agent with variational preferences with respect to $B_\delta(\P)$, i.e., the agent considers the maxmin criterion
$$ \sup_{\pi\in \cA} \inf_{\tilde \P\in B_\delta(\P)}\E_{\tilde \P}[u(\langle X,\pi\rangle)]$$
of \cite{gilboa1989maxmin} and picks their best strategy against the nature who picks the worst model from within the ball $B_\delta(\P)$. We are interested in understanding how this changes the classical problems considered in Section \ref{sec:baseline}. Specifically, we want to quantify how the value and the actions of the agent change for small $\delta$. In this section, we only give an overview of the main results and highlight some of their consequences. The rigorous statements require assumptions which, even if not surprising or restrictive, are often cumbersome to state. Consequently, we defer these to later sections. 

\subsection{Wasserstein variational preferences}\label{sec:WVP}

We start with a simple remark on variational preferences of \cite{gilboa1989maxmin}. 
Consider an agent using a strategy $\pi$ leading to wealth $\langle X,\pi\rangle$. The agent's utility function $u$ and their choice of $p$ and $\delta$ imply a preference relation on the space of models $\cP_p(\cS)$, i.e., distributions of $X$, via
\begin{align}\label{eq:Wass pref rel}
\P \succeq_W \check{\P} \quad \Leftrightarrow \quad \min_{\tilde{\P}\in B_\delta(\P)} \E_{\tilde{\P}}[u(\langle X, \pi \rangle)] \ge \min_{\tilde{\P}\in B_\delta(\check{\P})} \E_{\tilde{\P}}[u(\langle X, \pi \rangle)].
\end{align}
Assuming that $u$ is differentiable and applying \cite[Theorem 2]{bartl2020robust} we obtain that, up to $o(\delta)$,
$\P \succeq_W \check{\P}$ if and only if 
\begin{align*}
&\E_{\P}[u(\langle X, \pi\rangle))] - \delta |\pi|\left(\E_{\P}\left[|u^\prime(\langle X, \pi \rangle))|^q\right]\right)^{1/q} \ge \E_{\check\P}[u(\langle X, \pi\rangle)]  - \delta |\pi| \left(\E_{\check\P}\left[|u^\prime(\langle X, \pi\rangle)|^q\right]\right)^{1/q}.
\end{align*}
In particular these preferences first take the expectation of $u$ under $\P$ and $\check\P$ into account. If these align, then the norm of the first derivative of $u$, or the subjective pricing kernel (stochastic discount factor) is decisive for small $\delta>0$. We believe that the  fact that marginal quantities appear and control the assessment of small perturbations of the model seems both natural and desirable. 

To the best of our knowledge, \eqref{eq:Wass pref rel} is the first instance of using the Wasserstein balls to define preferences. Instead, in the economics literature, balls with respect to relative entropy have been suggested. \cite{hansen2001robust} call these \emph{constraint preferences} and observed that these are equivalent to using multiplier preferences with a penalty proportional to the relative entropy, see also \cite{maccheroni2006ambiguity}. Accordingly, if we consider
\begin{align}\label{eq:Wass pref rel1}
\P \succeq_{KL} \check{\P} \quad \Leftrightarrow \quad \min_{\tilde{\P}: H(\tilde \P|\P)\leq \delta^2} \E_{\tilde{\P}}[u(\langle X, \pi\rangle)] \ge \min_{\tilde{\P}: H(\tilde \P|\check{\P})\leq \delta^2} \E_{\tilde{\P}}[u(\langle X, \pi\rangle)]
\end{align}
then using the results of \cite{lam2016robust}, assuming finiteness of exponential moments, it follows that, up to $o(\delta)$, $\P\succeq_{KL}\check{\P}$ if and only if 
\begin{align*}
&\E_{\P}[u(\langle X, \pi\rangle ))] - \delta \sqrt{2\Var_\P(u(\langle X, \pi\rangle))} \ge \E_{\check\P}[u(\langle X, \pi\rangle)] - \delta \sqrt{2\Var_{\check{\P}}(u(\langle X, \pi\rangle))},
\end{align*}
where $\Var_\P(u(\xi))$ is the variance of $u(\xi)$ under $\P$. In particular, in contrast to $\succeq_W$, marginal quantities do not appear but instead the second moment of the utility $u(\xi)$ is decisive.

\subsection{Distributionally robust expected utility maximisation}\label{sec:REUM}

Consider a distributionally robust expected utility maximisation problem of the form
\begin{equation}\label{eq. minimax}
    V\left( \delta \right) : =\ \sup_{\pi\in \Adm} \inf_{\tilde{\P} \in B_{\delta}(\P)}\E_{\tilde{\P}}\left[u\left(\langle  X,\pi \rangle \right)\right].
\end{equation}
This problem can be written in a number of equivalent ways, using different representations of the $B_\delta(\P)$-ball. In particular, when $\P$ is non-atomic, denseness of Monge couplings among all couplings, see \cite{pratelli2007equality}, implies that
$$V(\delta) =\ \sup_{\pi\in \Adm} \inf_{f: \E_{\P}[|f(X)|^p]\leq \delta^p}\E_{\P}\left[u\left(\langle  X+f(X),\pi \rangle \right)\right].$$
For $\delta\ge 0$ we denote by $\pi^\opt_\delta\in \Adm^\opt_{\delta}$ the optimizers for $V(\delta)$. 
In particular, $V(0)=V$ and $\pi^\opt_0=\pi^\opt$ as defined in Section \ref{sec:EUM}. In Theorem \ref{thm:EUMsens} we show that
\begin{align*}
V'(0)=\lim_{\delta\to 0} \frac{V(\delta)- V(0)}{\delta} = -|\pi^\opt| \left( \E_{\P} \left[|u^\prime(\langle X, \pi^\opt \rangle)|^q\right] \right)^{1/q},
\end{align*}
where $q=p/(p-1)$. In particular, the value function always decreases when model misspecification occurs. This is intuitively obvious: the agent considers a max-min problem and hence is model-misspecification-averse. The loss in value is proportional to the number of shares held $|\pi^\opt|$ and to the norm of the agent's pricing kernel. The first accounts for the agent's total exposure to the market and the second for the relative distance between the physical measure and the agent's subjective pricing measure. We note also that $V'(0)$ is decreasing in $q$, i.e., is increasing in $p$. Again, this is intuitive: larger $p$ corresponds to a stronger metric and hence smaller balls $B_\delta(\P)$, i.e., less uncertainty about the baseline model. However,  $V'(0)$ does not appear to have any other simple monotonicity properties in relation to the baseline model. In particular, changing $\P$ to increase the Sharpe ratio or the value $V$ may lead both to increasing or decreasing $V'(0)$. Arguing purely formally, we expect that $\pi^\opt$ increases with the Sharpe ratio and thus the monotonicity of $V^\prime(0)$ depends on the tradeoff between $|\pi^\opt|$ and $|u^\prime(\langle X, \pi^\opt \rangle)|$, which is usually decreasing in $\pi^\opt$. This tradeoff can lead to different behaviour for different utility functions and baseline models, see section \ref{sec:examples} for examples.
 
Next, considering the optimal trading strategy, in Theorem \ref{thm:a_deriv} we show that
\begin{equation}\label{eq:sensitivity_optimiser}
    \pi^\opt(\delta) = \pi^\opt + \pip(0) \delta + o(\delta),
\end{equation}
where the gradient $\pip(0)$ is given by
\begin{align*}
 \pip(0)  
    =&\left(\nabla_\pi^2 V(0) \right)^{-1} \cdot \frac{\pi^\opt}{\left| \pi^\opt \right|}\cdot \kappa_u,\quad \text{with}\\
    &\kappa_u= \|u^\prime(\langle X , \pi^\opt \rangle)\|_{L^q(\P)}^{1-q} \cdot \left(\E_{\P}^{} \left[\frac{\langle X,\pi^\opt \rangle u^{\prime\prime} (\langle X,\pi^\opt \rangle)+u^\prime(\langle X, \pi^\opt \rangle) }{\left| u^\prime(\langle X, \pi^\opt \rangle) \right|^{1-q} }\right]  \right).
\end{align*}
The first two terms above decide about the relative adjustments to the components in $\pi^\opt$, while $\kappa_u$ is a constant multiplier. The first term is the inverse Hessian matrix, analogue to the inverse Fisher information matrix in a statistical problem, see \cite{bartl2020robust}. It is multiplied by the relative weights in the portfolio, the second term.

\subsection{Distributionally robust marginal utility price}\label{sec:mainresults_marginalprice}
We introduce now a robust version of the marginal utility price of \cite{davis1997option}. We recall that $g$, see Assumption \ref{Ass:stand}, denotes the payoff of the option we want to price. Define 
\begin{align}\label{eq:Vpd def}
V(\delta, \varepsilon, \d)= \sup_{\pi \in \Adm} \inf_{\tilde{\P} \in B_{\delta}(\P)} \E_{\tilde{\P}} \left[ u\left(-\epsilon+\langle X, \pi \rangle+\frac{\epsilon}{\d} g(X)\right) \right].
\end{align}
\begin{definition}
Suppose that for each $\d>0$ the function $\varepsilon \mapsto V(\delta,\epsilon, \d)$ is differentiable. A number  $\hatd(\delta)$, which satisfies
\begin{align*}
\nabla_\epsilon V(\delta,0, \hatd(\delta))=0.
\end{align*}
is called a robust marginal utility price for the  uncertainty level $\delta$.
\end{definition}
Note that for $\delta=0$ this notion agrees with the marginal utility price of \cite{davis1997option},  $\hatd(0)=\hatd$ from Definition \ref{def:marginal}, and can thus be considered as its natural distributionally robust counterpart. Furthermore, in Theorem \ref{thm. marginal} below we show that it is still computed via an expectation under a subjective martingale measure, only now this choice of measure also depends on the level of uncertainty $\delta$:
\begin{align*}
\hatd(\delta)=\frac{ \E_{\P^\opt}\left[ u^\prime (\langle X,\pi^\opt_\delta\rangle)\,g(X)\, \right]}{\E_{\P^\opt} \left[ u^\prime(\langle X, \pi^\opt_\delta\rangle)\right]},
\end{align*}
where $\P^\opt$ is a minimising measure in $B_\P(\delta)$ in \eqref{eq. minimax}, the nature's optimal response to the agent's best strategy $\pi^\opt_\delta$. Interestingly, the marginal price $\hatd(\delta)$ is not monotone in $\delta$, in particular the price can both increase or decrease as model uncertainty is introduced. The behaviour of $\hatd(\delta)$ is specific to the agent and the option's payoff. This again is intuitive: there is an interplay between an agent's trading intent and their valuation of the option, while uncertainty affects both. The only case when one would expect the marginal price to always decrease when uncertainty is introduced, is when the agent does not see trading as profitable to start with, i.e., when $\pi^\opt=0$. This is confirmed by our results on the first order sensitivity in $\delta$ given in Theorem \ref{thm:marg_sens}. We find that if $\pi^\opt=0$, then 
\begin{align*}
\hatd'(0)=-\left( \E_{\P} \left[ |\nabla g(x)|^q\right]\right)^{1/q}.
\end{align*}
If $\pi^\opt\neq 0$ this sensitivity is more involved and, in particular, can be both positive or negative. Remarkably, we can still compute it in a closed form:
\begin{equation*}
\hatd'(0) = \E_{\Q_u} \left[ R_u(\langle  X,\pi^\opt\rangle) \left( \langle T(X),\pi^\opt\rangle-\langle X,\pip(0)\rangle\right) \cdot \left(g(X)-\hatd\right)- \langle \nabla g (X) ,T(X)\rangle\right],
\end{equation*}
where $\Q_u$ was given in \eqref{eq:u mart mg}, $R_u(x)=-\frac{u''(x)}{u'(x)}$ is the agent's absolute risk aversion coefficient and 
\begin{align*}
T(x):=\frac{\pi^\opt}{|\pi^\opt|} |u^\prime (\langle x,\pi^\opt\rangle)|^{q-1} \left(\E_{\P} \left[ |u^\prime (\langle X,\pi^\opt \rangle)|^q\right] \right)^{1/q-1}.
\end{align*}

\section{Examples}\label{sec:examples}

In this section we consider a number of simple examples to illustrate the notions and results presented above. We selected the examples so that most of the computations can be derived in two ways: either through a direct brute force computation, see \cite{obloj2021}, or by using results presented in section \ref{sec:main results} and stated in more detail in sections \ref{sec:robust eum}-\ref{sec:robust Davis price} below.

Specifically, throughout this section we take $d=1$ and mainly focus on $\P$ which is either binomial or Gaussian. The former is the only complete model in this one-period setting so that, in particular, the martingale measure $\Q_u$ is unique and independent of the utility function $u$. Introducing model uncertainty however, we lose market completeness, and sensitivity of Davis' price is subjective. In all of the figures, we plot sensitivities as functions of the Sharpe ratio $\mu/\sigma$ of the baseline model $\P$, where $\mu=\E_\P[X]$ and $\sigma^2=\Var_{\P}(X)$. 

\subsection{Binomial model}

Fix $a\in (0,1/2)$. Consider the baseline model $\P=a\delta_{-1}+(1-a)\delta_1$ and a log investor with initial capital equal to one, so that $u(x)=\log(1+x)$. 
For concreteness we take $\mathcal{D}=(-1,\infty)$, $\Adm=(-1+a,1-a), \mathcal{S}=[-1-a,1+a]$ and note that for $\epsilon_0=a^2/2$ we have
\begin{align*}
\{\langle x, \pi \rangle:\ x\in \mathcal{S}, \pi\in \Adm \}^{\epsilon_0}\subset (-1+a^2, 1-a^2)^{\epsilon_0} = [-1+a^2/2, 1-a^2/2]\subseteq \mathcal{D}
\end{align*}
and that any continuous function $g$ is bounded on $\mathcal{S}$.
The unique optimiser $\pi^\opt$ is then given by $\pi^\opt=1-2a$ and $V=a\log(2a)+(1-a)\log(2-2a).$\\

We now consider the robust optimal investment problem $V(\delta)$. For the case $p=\infty$ we explicitly calculate 
\begin{align*}
\pi^\opt_\delta&=a \log \left( \frac{2a}{1-\delta}\right) +(1-a)\log \left( \frac{-2+2a}{-1-\delta}\right),
\end{align*}
which in turn yields
\begin{align*}
V^\prime(0)&=-\pi^\opt.
\end{align*}
The same result follows directly from Theorem \ref{thm:EUMsens}, which also covers the case of finite $p$. We can compare this to  \cite{lam2016robust}, where uncertainty is quantified by balls in KL-divergence. Indeed, an application of \cite[Theorem 3.1]{lam2016robust} yields
\begin{align*}
\tilde{V}^\prime(0)
&=-\sqrt{2 \text{Var}(\log(1+\pi^\opt X))}\\
&=-\Bigg(2 \Big( a\log(2a)^2+(1-a)\log(2(1-a))^2\\
&\qquad\qquad -\Big( a\log(2a)+(1-a)\log(2(1-a))\Big)^2 \Bigg)^{1/2},
\end{align*}
for 
\begin{align*}
\tilde{V}(\delta):= \sup_{\pi\in \Adm} \inf_{H(\tilde \P|\P)\leq \delta^2}\E_{\tilde \P}\left[u\left(\langle  X,\pi \rangle \right)\right],
\end{align*}
see Figure \ref{fig:2} for a comparison.

\begin{figure}[h!]
\centering 
\begin{minipage}[b]{0.49\textwidth}
  \includegraphics[width=1\textwidth]{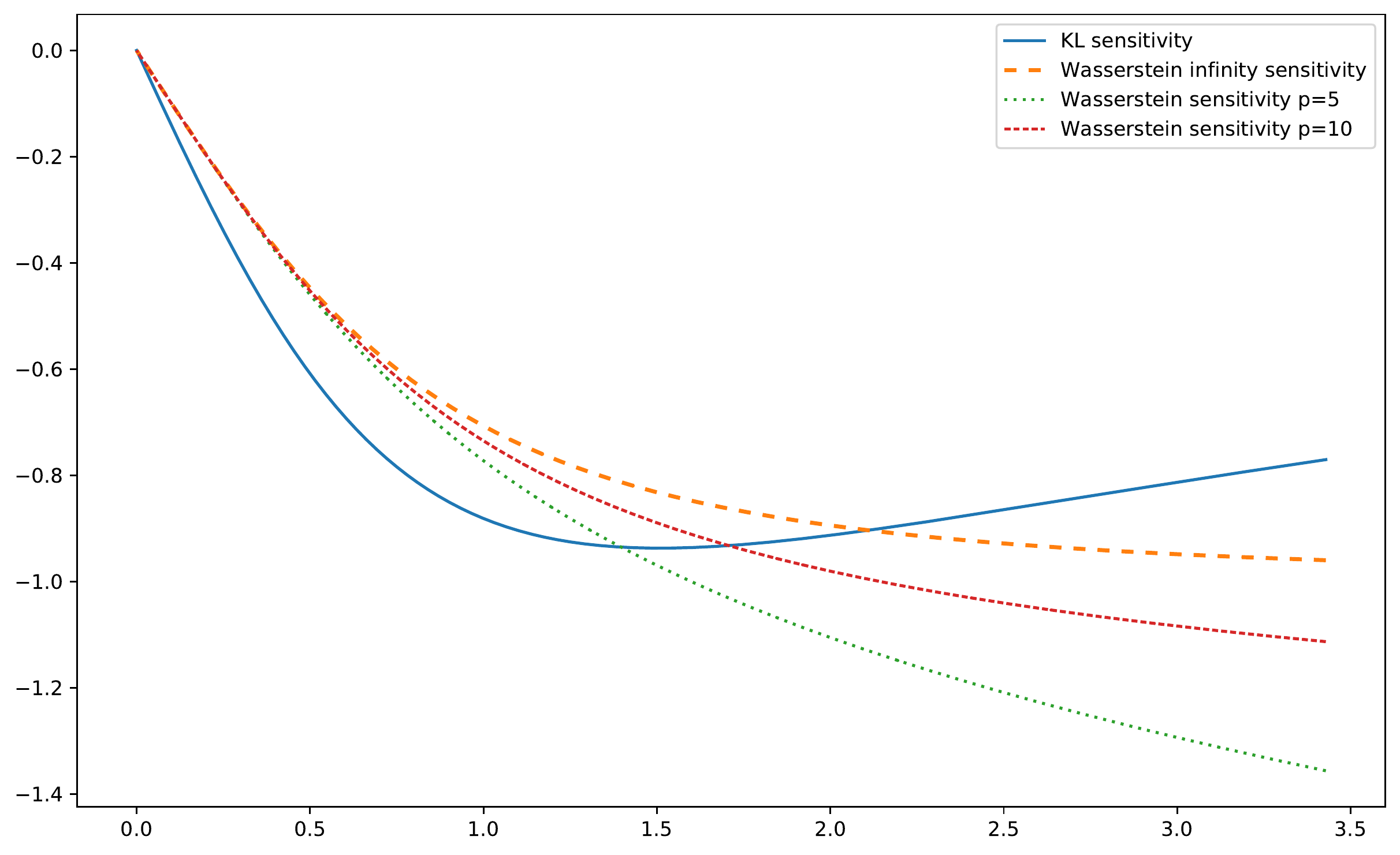}
\end{minipage}
\hfill
\begin{minipage}[b]{0.49\textwidth}
  \includegraphics[width=1\textwidth]{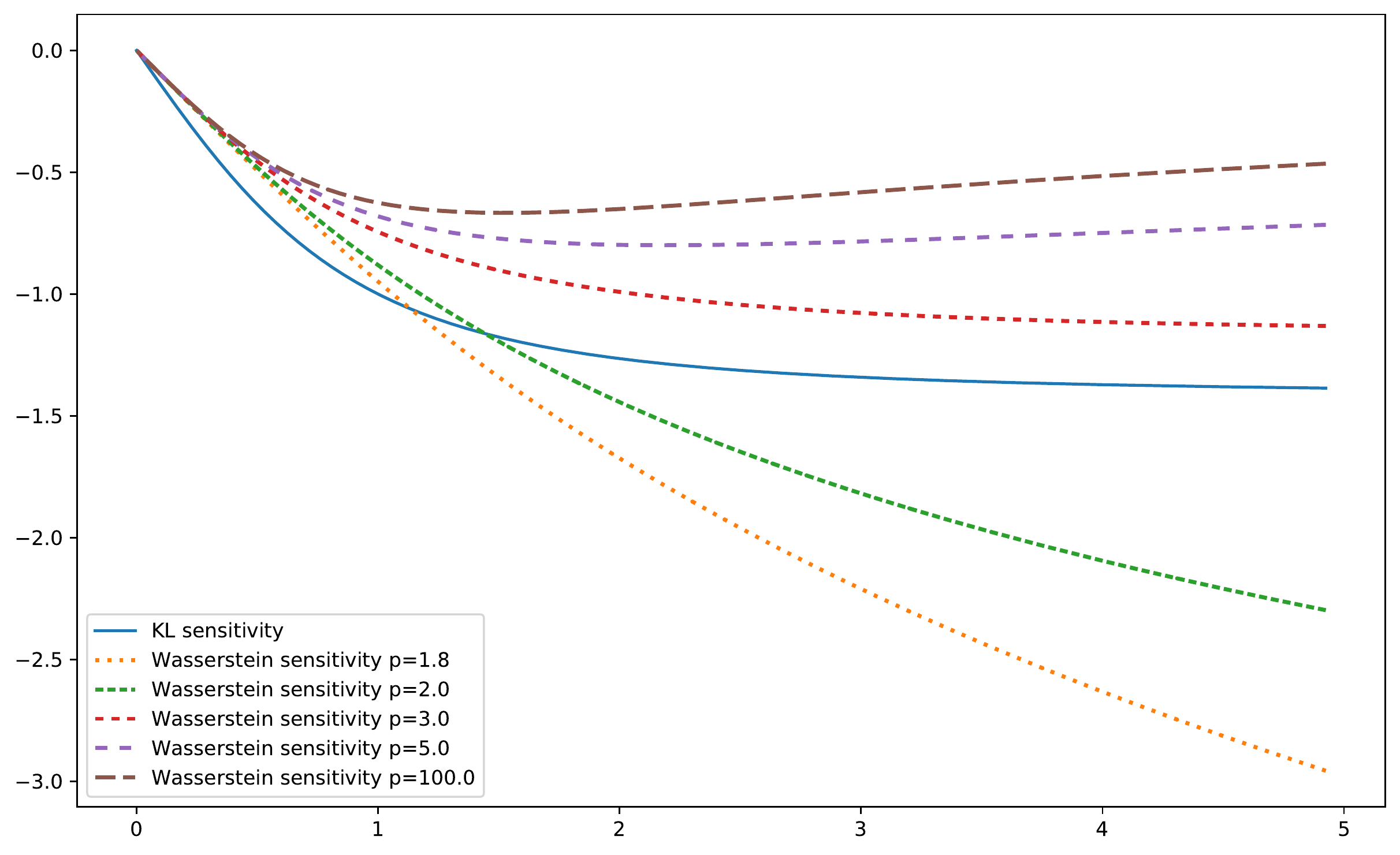}
\end{minipage}
\label{fig. conc3d}
 \captionsetup{width=\linewidth}
\caption{Wasserstein-sensitivity $V^\prime(0)$ as a function of $\mu/\sigma$ for different $p$ and KL-sensitivity $\tilde{V}^\prime(0)$ for $u(x)=\log(1+x)$ (left) and $u(x)=-\exp(-x)$ (right).}
\label{fig:2}
\end{figure}


This plot raises the question if $V^\prime(0)$ is generally a decreasing function of the Sharpe ratio $\mu/\sigma$ for all $p$ in the binomial model. This turns out not to be the case. Indeed if we choose the exponential utility function $u(x)=-\exp(-\gamma x)$ for $\gamma>0$, then 
\begin{align*}
\pi^\opt= \frac{\log\left(\frac{a}{1-a} \right)}{-2\gamma}
\end{align*}
and in particular 
\begin{align*}
V^\prime(0)
&= -\left(a \left(\frac{a}{1-a}\right)^{-q/2} + (1-a)\left( \frac{a}{1-a}\right)^{q/2}\right)^{1/q}\cdot  \pi^\opt.
\end{align*}
Clearly, $V^\prime(0)_{|a=\frac{1}{2}}=0$, as expected since $\pi^\opt_{|a=\frac{1}{2}}=0$. Consider now the asymptotics as $a\to 0$. Clearly $\pi^\opt\to \infty$. Considering the leading behaviour for the first term (in parentheses), we see that for $q\geq 2$ it diverges so that $V^\prime(0)\to -\infty$, but for $1<q<2$ it converges to zero and dominates so that $V^\prime(0)\to 0$. In particular, we see that the monotonicity of $V^\prime(0)$ depends on $p$ and, for $p>2$, $V^\prime(0)$ is in fact increasing for large Sharpe ratio $\mu/\sigma$, see Figure \ref{fig:2}. Finally, for comparison, we note that \cite[Theorem 3.1]{lam2016robust} yields
\begin{align*}
\tilde{V}^\prime(0)&=-\sqrt{2 \left(1- 4a(1-a)\right)}.
\end{align*} 

We continue our investigation with the sensitivity of the optimiser $\pi^\opt(\delta)$: a direct calculation yields $\pip(0)=-1$. Alternatively we can use Theorem \ref{thm:a_deriv} to obtain
\begin{align*}
\pip(0)=-a(1-a)\left(a^{-q+1}+(1-a)^{-q+1}\right)^{1/q-1}\cdot(a^{-q}+(1-a)^{-q}),
\end{align*}
so that the results coincide for $q=1$. See Figure \ref{fig:3} for a plot of $\pip(0)$ for different values of $p$.

\begin{figure}[h!]
\centering 
\begin{minipage}[b]{0.49\textwidth}
  \includegraphics[width=1\textwidth]{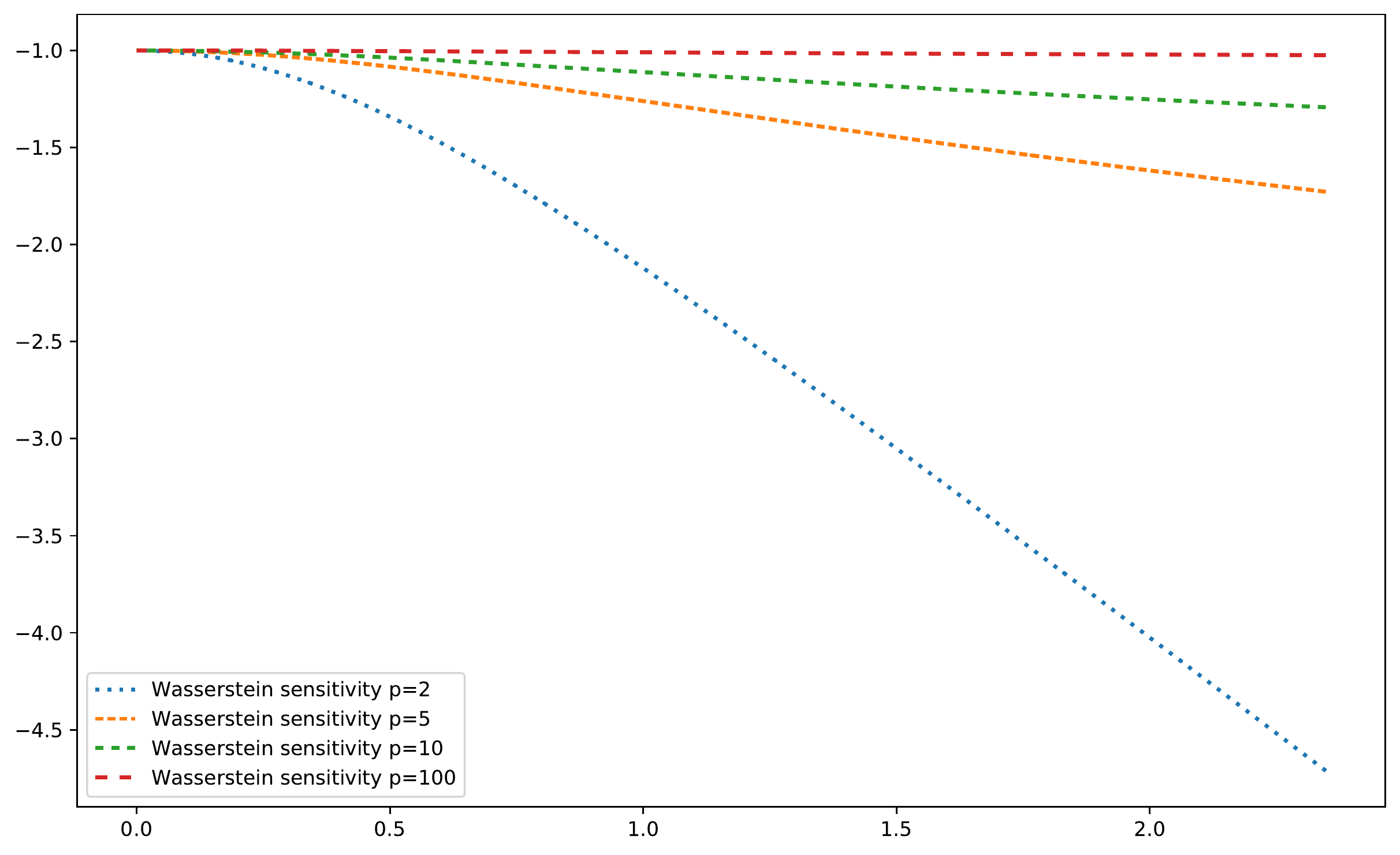}
\end{minipage}
\hfill
\begin{minipage}[b]{0.49\textwidth}
  \includegraphics[width=1\textwidth]{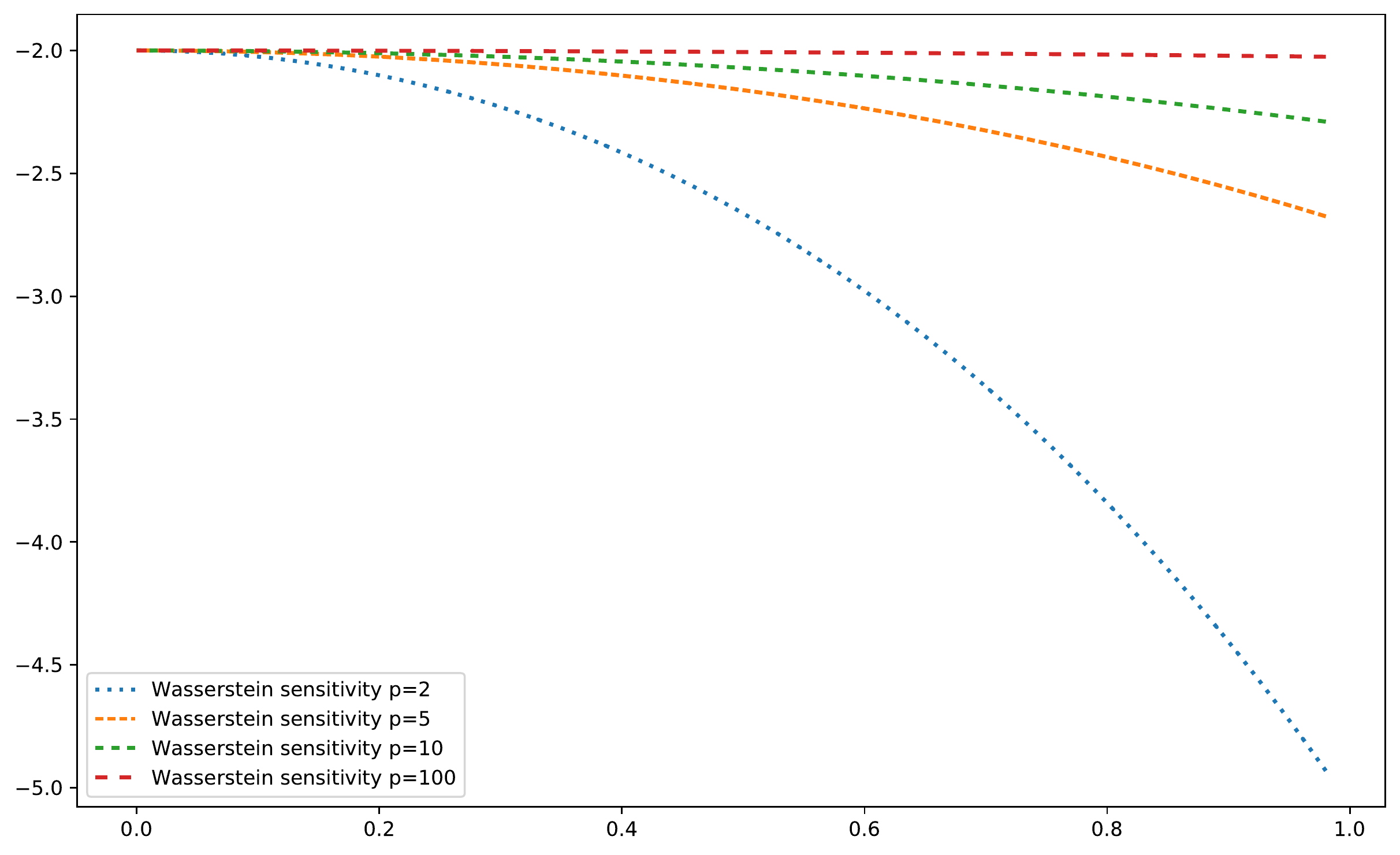}
\end{minipage}
\label{fig. conc3d}
 \captionsetup{width=\linewidth}
\caption{Wasserstein-sensitivity  $\pip(0)$ of $\pi^\opt_\delta$ (left) and Wasserstein-sensitivity $\hatd^\prime(0)$ of the Davis price $\hatd$ (right) for different $p$. Both are plotted as functions of $\mu/\sigma$.}
\label{fig:3}
\end{figure}

Let us now compare the Davis price of the baseline model with its robust counterpart. According to \eqref{eq:davis_main} it is given by $
\hatd=\E_{\Q_u}[g(X)]$. We consider first $g(x)=x^3$ as a concrete example. An easy computation shows $\hatd=0$, while we can explicitly calculate the robust Davis price $\hatd(\delta)$ as 
\begin{align*}
\hatd(\delta)=-2\delta+2\delta^3.
\end{align*}
In particular, $\hatd^\prime(0)=-2$ as can also be derived from Theorem \ref{thm:marg_sens}. For a general $p\in (1, \infty]$ we obtain
\begin{align*}
(\hatd)^\prime(0)&=-\frac{1}{4a} \left((2a)^{1-q} (1-2a) \cdot 2^{q-1}\left(a^{-q+1}+(1-a)^{-q+1}\right)^{1/q-1} -1\right) \\
&\quad + \frac{1}{4(1-a)} \left( (2(1-a))^{1-q} (1-2a) \cdot 2^{q-1}\left(a^{-q+1}+(1-a)^{-q+1}\right)^{1/q-1} +1\right)\\
&\quad -\frac{3}{2} (a^{1-q}+(1-a)^{1-q})\cdot \left(a^{1-q}+(1-a)^{1-q}\right)^{1/q-1},
\end{align*}
see Figure \ref{fig:3}.
Next consider $g(x)=x^+$ and $p=\infty$. Note that $g$ is not differentiable, but as $\P$ is only supported in $\{-1,1\}$ we can take a suitable $C^\infty$-approximation instead. 
We compute $\hatd=0.5$ and
\begin{align*}
\hatd(\delta)=\frac{1-\delta^2}{2}.
\end{align*}
In particular $\hatd^\prime(0)=0$, which is also readily seen using Theorem \ref{thm:marg_sens}. 

As we remarked before, it is however not always true that $\hatd(\delta)\le \hatd$. Take, e.g., $g(x)=|x+x_0|$ for some $x_0\in (0,1)$. For $p=\infty$, as $\hatd(\delta)$ is the expectation under the unique martingale measure concentrated on $-1-\delta$ and $1-\delta$, we conclude that
\begin{align*}
\hatd(\delta)=1-\delta^2+\delta x_0 > 1=\hatd
\end{align*}
for small enough $\delta>0$.

\subsection{Normal model}\label{sec:normal}

We now set $\P=\mathcal{N}(\mu,\sigma^2)$ and consider the utility function $u(x)=-e^{-\gamma x}$ for $\gamma>0$ and $\mathcal{D}=\R$. We obtain
$$\pi^\opt=\frac{\mu}{\gamma \sigma^2}\quad \text{and}\quad V=-\exp\left(-\frac{\mu^2}{2\sigma^2}\right).
$$
For $p=\infty$ we find 
\begin{align*}
V(\delta)=-\exp\left(-\frac{(\mu-\delta)^2}{2\sigma^2}\right).
\end{align*}
A direct calculation thus gives
\begin{align*}
V^\prime(0)=\exp\left(-\frac{\mu^2}{2\sigma^2}\right)\frac{\mu}{\sigma^2},
\end{align*}
which can be recovered by Theorem \ref{thm:EUMsens}. Figure \ref{fig:1} compares the sensitivity $V^\prime(0)$ as a function of $\mu/\sigma$ for different $\mu$. We also remark that in this case we can not compare with model uncertainty in the sense of relative entropy balls as this problem is degenerate. In fact 
 \begin{align*}
&\E_{\P}\left[\exp(\theta \left( u(\langle X,\pi^\opt\rangle)\right) \right]=\E_{\P}\left[\exp\left(\theta \exp\left(\frac{\mu}{\sigma^2} X\right)\right)\ \right]=\infty
\end{align*}
for $\theta,\mu\neq 0$, so that \cite[Assumption 3.1]{lam2016robust} is not satisfied.
\begin{figure}[h!]
\includegraphics[scale=0.4]{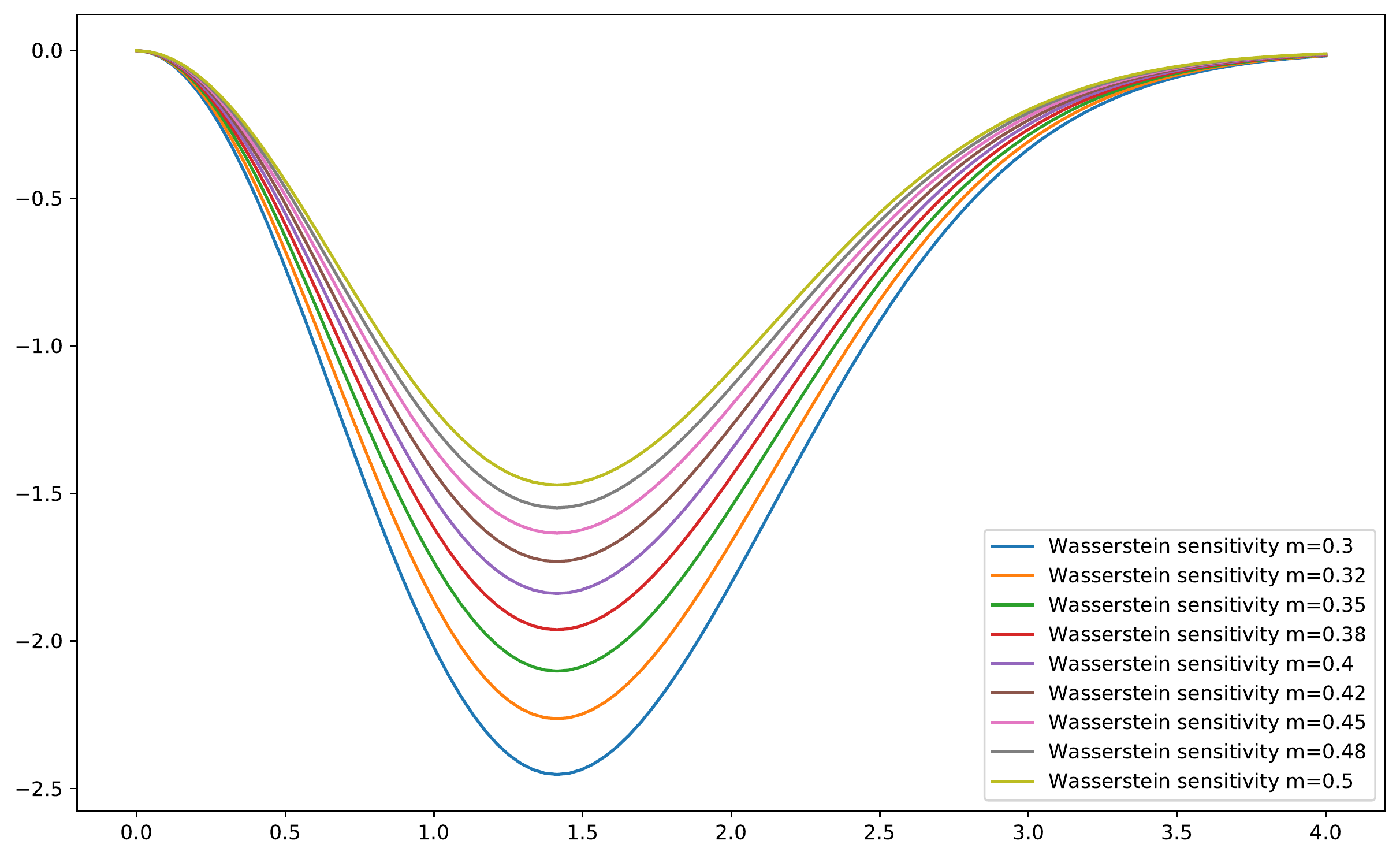}
\caption{Wasserstein-$\infty$ sensitivity $V^\prime(0)$ as a function of $\mu/\sigma$ for different $\mu$.}
\label{fig:1}
\end{figure}
Next we find that the distributionally robust optimiser is
\begin{align*}
\pi^\opt_\delta =\frac{\mu-\delta}{\gamma\sigma^2},\quad \text{and thus }\pip(0)=-\frac{1}{\gamma\sigma^2},
\end{align*}
which can again be recovered by Theorem \ref{thm:a_deriv}.

We can also compute the Davis price explicitly. E.g., for the case $g(x)=x^2$ we obtain $\hatd=\sigma^2$.
Interestingly enough this result remains unchanged for any positive $\delta>0$, i.e., $\hatd(\delta)=\sigma^2$. More generally one can show that 
\begin{align*}
\hatd(\delta)=\E_{\dot{\P}}[g(X)]
\end{align*}
for any $\delta>0$, where $\dot{\P}=\mathcal{N}(0,\sigma^2)$. This implies in particular $\hatd^\prime(0)=0$ for any payoff $g$, a result which can again be recovered from Theorem \ref{thm:marg_sens}, albeit through a tedious calculation, see  \cite[Section 1]{obloj2021}.

\subsection{Discussion of normal model in the case $p\in (1,\infty)$}

Note that in Section \ref{sec:normal} the exponential function $u(x)=-e^{-\gamma x}$ does not satisfy Assumption \ref{Ass:1}.\textit{(i)} for any $p\in (1,\infty)$, so Theorem \ref{thm:EUMsens} is not applicable. Indeed it is not hard to see that $V^\prime(0)=-\infty$ in this case: this follows by noting that for any $\pi\neq 0$
\begin{align*}
W_p(\delta_{-\text{sign}(\pi)\cdot n}, \P)\le n+\left(\E_{\P}[|X|^p]\right)^{1/p}=:f(n),
\end{align*} 
where $f(n)=O(n)$. In particular for all $n\in \N$ large enough, the probability measures $$\tilde{\P}^n=\left(\frac{\delta}{f(n)}\right)^p\delta_{-\text{sign}(\pi)\cdot n}+\left(1-\left(\frac{\delta}{f(n)}\right)^p\right)\P$$ satisfy
\begin{align*}
W_p(\tilde{\P}^n, \P)\le \frac{\delta}{f(n)} W_p(\delta_{-\text{sign}(\pi)\cdot n}, \P) \le \delta
\end{align*}
so that $\tilde{\P}^n \in B_{\delta}(\P)$. Thus 
\begin{align*}
\inf_{\tilde{\P}\in B_{\delta}(\P)} \E_{\tilde{\P}}[u(\langle X, \pi\rangle)] &\le \E_{\tilde{\P}^n} [u(\langle X, \pi\rangle]\\
&= -\left(\frac{\delta}{f(n)}\right)^p \exp(\gamma |\pi|n)-\left(1-\left(\frac{\delta}{f(n)}\right)^p\right) \E_{\P}[\exp(-\gamma\pi X)],
\end{align*}
which goes to $-\infty$ with $n\to \infty$, using $f(n)=O(n)$. In particular 
\begin{align*}
V(\delta)&= \sup_{\pi \in \Adm}\inf_{\tilde{\P}\in B_{\delta}(\P)} \E_{\tilde{\P}}[u(\langle X, \pi\rangle)]=\sup_{\pi \in \Adm}\inf_{\tilde{\P}\in B_{\delta}(\P)} \E_{\tilde{\P}}[-\exp\left(-\gamma\langle X, \pi\rangle\right)]\\
&=\E_{\tilde{\P}}[-\exp(-\gamma \langle X, 0\rangle)]=-1
\end{align*}
and so $$V^\prime(0)=\lim_{\delta\to 0}\frac{V(\delta)-V(0)}{\delta}= \lim_{\delta\to 0} \frac{-1+\exp\left(-\frac{\mu^2}{2\sigma^2}\right)}{\delta}=-\infty$$
for $\mu\neq 0$.
As we have seen above, Wasserstein-$p$-balls do not allow for enough control over the tails of the distribution when considering a utility function decreasing exponentially. There are two ways to remedy this:
\begin{enumerate}[(i)]
\item Consider an approximating sequence of utility functions $(u_\kappa)_{\kappa \in (0,1)}$ satisfying Assumption \ref{Ass:1}, such that $u_\kappa\to u$ for $\kappa\to 0$.
\item Use a different Wasserstein distance adapted to the utility function under investigation.
\end{enumerate}
We will briefly comment on both approaches.
For (i) we can formally consider $$u_\kappa(x):=\mathds{1}_{\{ x< -1/\kappa\}} \left[-e^{\gamma/\kappa}+\gamma e^{\gamma/\kappa}(x+1/\kappa)\right]+\mathds{1}_{\{x\ge -1/\kappa\}} u(x)$$ for $\kappa> 0$ and note that $u_\kappa(x)\ge u(x)$ for all $x\in \mathcal{S}$, so that $$u(x)=\inf_{\kappa> 0}u_\kappa(x)=\lim_{\kappa\to 0}u_\kappa(x).$$ 
In particular $(u_\kappa)^\prime(x)\le \gamma e^{\gamma/\kappa}$ and $(u_{\tilde\kappa})^\prime(x)\ge (u_\kappa)^\prime(x)$ for all $\tilde\kappa\le \kappa$ and all $x\in \mathcal{S}$. Thus defining 
$$V_\kappa(\delta)=\sup_{\pi \in \Adm}\inf_{\tilde{\P}\in B_{\delta}(\P)} \E_{\tilde{\P}}[u_\kappa(\langle X,\pi\rangle)]$$
 we can apply the monotone convergence theorem to obtain
\begin{align*}
\lim_{\kappa\to 0} V^\prime_\kappa(0)&= \lim_{\kappa\to 0} -\left( \E_{\P} \left[|(u_{\kappa})^\prime(\langle X, \pi^\opt \rangle)|^q\right] \right)^{1/q} |\pi^\opt|\\
&=-\gamma\left( \E_{\P}\left[\exp(-\gamma\pi^\opt X)^{q}\right] \right)^{1/q} |\pi^\opt|\\
&=-\gamma\left( \exp\left( -\mu q \gamma\pi^\opt+\frac{\sigma^2 (q\gamma \pi^\opt)^2}{2} \right) \right)^{1/q}|\pi^\opt|\\
&=-\gamma  \exp\left( -\mu \gamma\pi^\opt+\frac{\sigma^2 q(\gamma\pi^\opt)^2}{2} \right)\left(\frac{\mu}{\gamma \sigma^2}\right)\\
&=-  \exp\left( -\frac{\mu^2}{\sigma^2}+\frac{ q \mu^2}{2\sigma^2} \right)\left(\frac{\mu}{\sigma^2}\right).
\end{align*}
This formula aligns with section \ref{sec:normal} for the case $q=1$. 
Similarly, we can also calculate 
\begin{align*}
\pip(0)= -\frac{\gamma^{1/q+q-3}}{\sigma^2}\exp\left( \frac{q-1}{2} \right) \left(-\frac{\mu^2(1-q)}{\sigma^2}+1\right),
\end{align*}
which again matches the result for $q=1$ in section \ref{sec:normal}, see Figure \ref{fig:4}.
\begin{figure}[h!]
\includegraphics[scale=0.4]{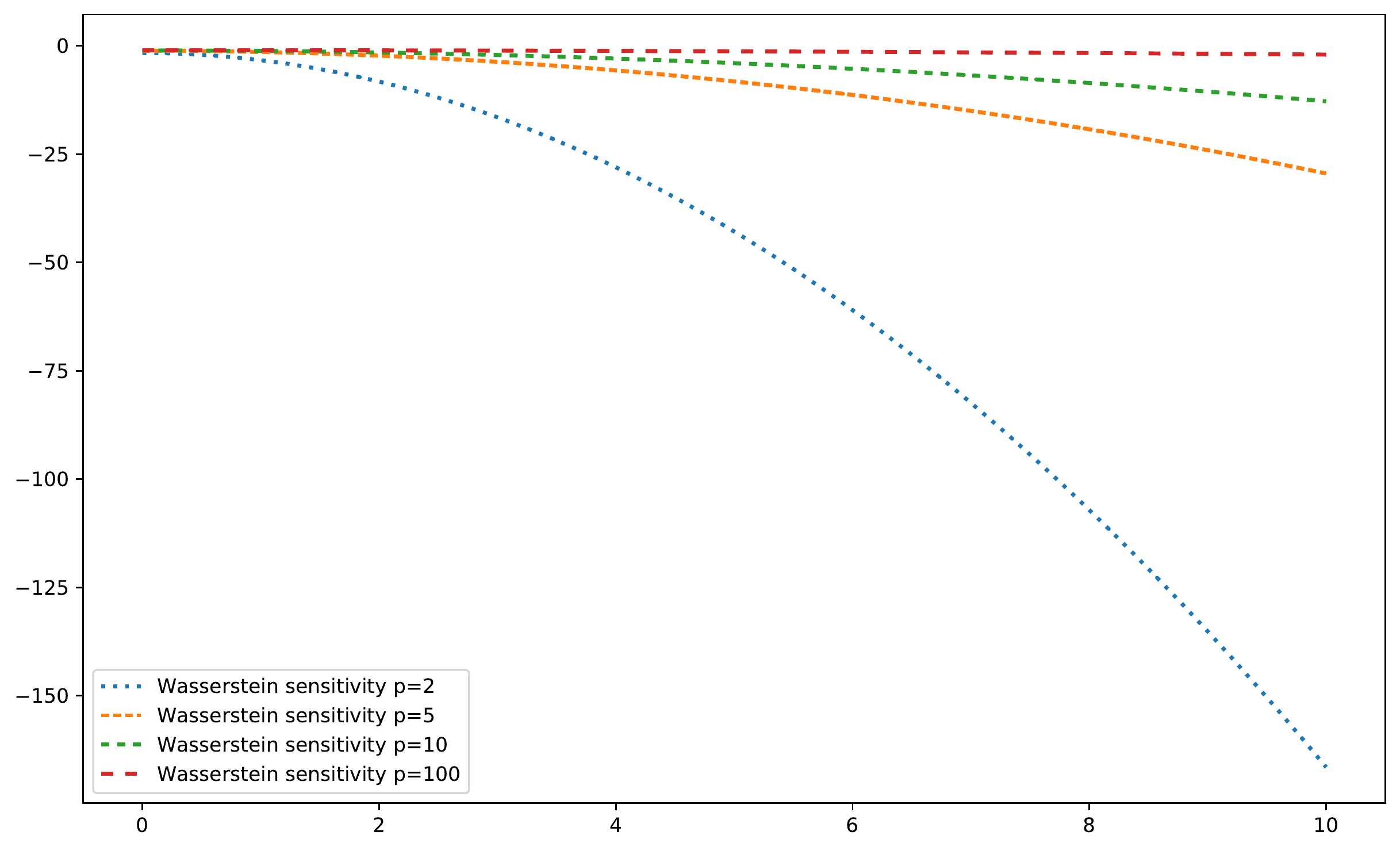}
\caption{Wasserstein-sensitivity $\pip(0)$ of $\pi^\opt_\delta$ as a function of $\mu/\sigma$ for $\sigma=\gamma=1$.}\label{fig:4}
\end{figure}
\\ Turning to (ii), assuming $u(0)=0$ we could work on generalised Orlicz hearts instead and define the Wasserstein-Luxemburg metric 
\begin{align*}
W_u(\P,\tilde{\P})&:=\inf\left\{ \int  u(|x-y|)\,\pi(dx,dy):\ \pi\in \mathrm{Cpl}(\P,\tilde{\P})\right\}\\
&=\inf\left\{ b>0:\ \int  \frac{u(|x-y|)}{b}\,\pi(dx,dy)\le 1, \text{where } \pi\in \mathrm{Cpl}(\P,\tilde{\P})\right\}.
\end{align*}
As $\|x\|:=u(|x|)$ is a norm, $W_u$ is still metric. A more general version of this definition was investigated in \cite{sturm2011generalized}. We further remark that the above definition can be connected to results of \cite{frittelli2002putting, cheridito2009risk}. From \cite[Theorem 2]{bartl2020robust} we know that
\begin{align*}
V^\prime(0)= -\| u^\prime (\langle X,\pi^\opt \rangle)\|_{u,\P}^\opt |\pi^\opt|,
\end{align*}
where $\|f \|_{u,\P}^\opt$ denotes a dual norm of $f$, which can be construed as a conjugate Orlicz norm of $f$ with respect to the probability measure $\P$. Alas, we were not able to compute either of these norms explicitly, even for the simple benchmark cases described above.\\

\subsection{Lognormal model}

Let us lastly consider the shifted lognormal distribution $\P=(1+x)_\#\exp\left(\mathcal{N}(\mu,\sigma^2)\right)$ with $\mu<-\sigma^2/2$ as well as $\Adm=[0,1]$, $\mathcal{D}=(-1,\infty)$ and $u(x)=\log(x+1)$. Here $(1+x)_{\#}\exp\left(\mathcal{N}(\mu,\sigma^2)\right)$ denotes the push forward of the distribution $\exp\left(\mathcal{N}(\mu,\sigma^2)\right)$ through the function $x\mapsto (1+x)$. In this case $\pi^\opt=0$ and for the butterfly payoff $g(x)=(x+K)^+-2x^++(x-K)^+$ with strike $K>0$ we compute the distributionally robust Davis price directly as
\begin{align*}
\hatd(\delta)&=\mathrm{BS}(1, 1, -K+\delta, \mu, \sigma)-2\mathrm{BS}(1, 1, \delta, \mu, \sigma)+\mathrm{BS}(1, 1, K+\delta, \mu, \sigma),
\end{align*}
where $\delta \geq 0$ and $\mathrm{BS}(T-t, S_0, K, \mu, \sigma)$ is the Black-Scholes call price. In consequence, $\hatd^\prime(0)$ corresponds to the partial derivative of $\mathrm{BS}(1, 1, -K, \mu, \sigma)-2\mathrm{BS}(1, 1, 0, \mu, \sigma)+\mathrm{BS}(1, 1, K, \mu, \sigma)$ w.r.t.\ $K$.

\section{Distributionally robust expected utility maximisation}\label{sec:robust eum}

We return now to the discussion of our main results and present rigorous statements of the results discussed in section \ref{sec:REUM} above.

\subsection{The value function and its sensitivity analysis}
In order to quantify the first-order sensitivity of \eqref{eq. minimax}, we start our analysis by calculating the sensitivity of the value function $V(\delta)$ using the general results obtained in \cite{bartl2020robust}. Remark that $\Adm^\opt_{0}=\{\pi^\opt\}$ follows simply from Assumption \ref{Ass:4a}. We make the following further assumption, where we recall that  $\interior{\Adm}$ denotes the interior of $\Adm$:

\begin{assumption}\label{Ass:1} The following hold:
\begin{enumerate}[(i)]
\item 
\begin{enumerate}
\item If $p<\infty$, then for every $r>0$ there exists $c>0$ such that
\begin{align*}
u^\prime (\langle x,\pi \rangle)\le c(1+|x|^{p-1})
\end{align*}
for all $x\in \mathcal{S}$ and $\pi\in \Adm$ with $|\pi|\le r$.
\item If $p=\infty$, then there exists $\tilde{\delta}>0$ such that for each $r>0$
\begin{align*}
 \E_{\P} \left[\sup_{\pi\in \Adm, |\pi|\le r} u^\prime \left( \langle X, \pi \rangle-r\tilde\delta\right)\right]<\infty.
\end{align*}
\end{enumerate}
\item  The optimiser $\pi^\opt\in \Adm^\opt_{0}$ satisfies  $\pi^\opt \in  \interior{\Adm}$.
\end{enumerate} 
\end{assumption}
In particular, we see that agent's ambiguity aversion, as captured by the choice of $p$, and their risk aversion, as captured by $u$, have to be compatible, see also Remark \ref{rem:p}. 
We obtain the following result:
\begin{theorem}\label{thm:EUMsens}
Suppose the utility function $u$ and the baseline model $\P$ satisfy Assumption \ref{Ass:1}. Then 
\begin{align*}
V'(0)=\lim_{\delta\to 0} \frac{V(\delta)- V(0)}{\delta} = -\left( \E_{\P} \left[|u^\prime(\langle X, \pi^\opt \rangle)|^q\right] \right)^{1/q} |\pi^\opt|.
\end{align*}
\end{theorem}

\begin{proof}[Proof of Theorem \ref{thm:EUMsens}]
By Assumptions \ref{Ass:4a} and \ref{Ass:1}, we have that the optimizer $\pi^\opt$ is unique, $\Adm^\opt_0=\{\pi^\opt\}$, and belongs to the interior of $\Adm$, $\pi^\opt\in \interior{\Adm}$. \\

\textit{Step 1:} In the case $p\in (1,\infty)$ we can simply apply \cite[Theorem 2]{bartl2020robust} for the function $f(x,\pi)=u(\langle x, \pi\rangle)$ to obtain
\begin{align*}
V^\prime(0)&= - \left( \E_{\P}\left[|\nabla_x f(X, \pi^\opt)|^q\right] \right)^{1/q}= - \left( \E_{\P}\left[|u^\prime(\langle X, \pi^\opt \rangle)|^q\right] \right)^{1/q} |\pi^\opt|.
\end{align*}

\textit{Step 2:} If $p=\infty$, then following \cite[proof of Theorem 2]{bartl2020robust} line by line and replacing \cite[equation (7)]{bartl2020robust} by
\begin{align*}
\sup_{\tilde{\P}\in B_{\delta}(\P)} \E_{\tilde{\P}} \left[ u^\prime ( \langle X, \pi^\opt \rangle) \right]
&\le \E_{\P} \left[ \sup_{|a|\le \delta} u^\prime ( \langle X+a, \pi^\opt \rangle ) \right]\\
&=\E_{\P}\left[ u^\prime ( \langle X, \pi^\opt \rangle-|\pi^\opt|\delta)\right]\le \E_{\P}\left[ u^\prime ( \langle X, \pi^\opt \rangle-r\delta)\right]<\infty
\end{align*}
for $|\pi^\opt|\le r$
we realise that 
\begin{align*}
\lim_{\delta\to 0}\E_{\pi^{\delta}} \left[ | u^\prime( \langle X+t(Y-X),\pi^\opt)|\right] = \E_{\P} \left[ | u^\prime( \langle X,\pi^\opt)|\right] 
\end{align*}
still holds by the dominated convergence theorem. This concludes the proof of the ``$\le$" inequality. For the ``$\ge$" inequality we can again follow the steps of \cite[proof of Theorem 2]{bartl2020robust}. In particular we note that
$
T(x)=T= \text{sign}\left( \pi^\opt\right) 
$ 
so that $\|T\|_{\infty}\le 1$ and again by the dominated convergence theorem we conclude that
\begin{align*}
\lim_{\delta\to 0} \E_{\P}\left[ u^\prime (\langle X-t\delta T(X), \pi^\opt_\delta\rangle) \langle \pi^\opt_\delta, T(X)\rangle \right]
&=\E_{\P}\left[ u^\prime (\langle X, \pi^\opt)  \right]\cdot \langle \pi^\opt, T\rangle\\
&=\E_{\P}\left[ u^\prime (\langle X, \pi^\opt) \right] \cdot  | \pi^\opt|,
\end{align*}
which concludes the proof.
\end{proof}

\subsection{The optimising strategy and its sensitivity analysis}\label{sec:optim}

Next, we  calculate the sensitivity of optimisers $\pi^\opt(\delta)\in \Adm^\opt_\delta$. In order to carry this out, we impose the following additional assumptions:

\begin{assumption}\label{Ass:2} 
The following hold:
\begin{enumerate} [(i)]
\item The function $u:\R\to \R$ is twice continuously differentiable and 
\begin{enumerate}[(a)]
\item if $p<\infty$,
$$ |x| |u^{\prime \prime} (\langle x, \pi\rangle)|+|u^\prime (\langle x, \pi\rangle)|\le c (1+|x|^{p-1-\epsilon})$$ for some $\epsilon>0$, $c>0$ and all $\pi\in \mathcal{A}$ close to $\pi^\opt$ and all $x\in \mathcal{S}$.
\item if $p=\infty$, then there exists $\tilde{\delta}>0$ such that 
\begin{align*}
\E_{\P}\left[\sup_{|z|\le \tilde \delta, |\pi|\le r}  |x+z| |u^{\prime \prime} (\langle x+z, \pi\rangle)|+|u^\prime (\langle x+z, \pi\rangle)| \right]<\infty.
\end{align*}
for all $\pi\in \mathcal{A}$ close to $\pi^\opt$ and all $x\in \mathcal{S}$.
\end{enumerate}

\item The matrix $$\nabla_\pi^2 V(0, \pi^\ast)=\E_\P\left[ XX^\top u^{\prime\prime}(\langle X, \pi^\opt\rangle)\right]$$ is negative definite.
\end{enumerate}
\end{assumption}

With these additional assumptions at hand, we can state the main result of this section as follows:

\begin{theorem}\label{thm:a_deriv}
Suppose the utility function $u$ and the baseline model $\P$ satisfy Assumptions \ref{Ass:1} and \ref{Ass:2}, and that $\pi^\opt\neq 0$. Take $\pi^\opt(\delta)\in \Adm^\opt_\delta$. Then 
\begin{equation*}
    \pi^\opt(\delta) = \pi^\opt + \pip(0) \delta + o(\delta),
\end{equation*}
where
\begin{align*}
\pip(0)  %
    &=\|u^\prime(\langle X , \pi^\opt \rangle)\|_{L^q(\P)}^{1-q}\cdot \left(\nabla_\pi^2 V(0) \right)^{-1} \cdot \frac{\pi^\opt}{\left| \pi^\opt \right|} \cdot \left(\E_{\P}^{} \left[\frac{\langle X,\pi^\opt \rangle u^{\prime\prime} (\langle X,\pi^\opt \rangle)+u^\prime(\langle X, \pi^\opt \rangle) }{\left| u^\prime(\langle X, \pi^\opt \rangle) \right|^{1-q} }\right]  \right).
\end{align*}
\end{theorem}

\begin{proof}[Proof of Theorem \ref{thm:a_deriv}]
\textit{Step 1:} Let us first consider the case $p\in (1, \infty)$. We note that uniqueness of $\pi^\opt$ and \cite[Lemma 19]{bartl2020robust} implies that $\pi^\opt(\delta) \to \pi^\opt$ as $\delta \to 0$. Furthermore Assumptions \ref{Ass:1} and \ref{Ass:2} are sufficient to apply \cite[Theorem 4]{bartl2020robust}. We thus conclude
\begin{equation*}
    \pi^\opt(\delta) = \pi^\opt + \pip(0) \delta + o(\delta)
\end{equation*}
where
\begin{align*}
    \pip(0)
    &= \left( \E_{\P}\left[ \left| \nabla_x f(X,\pi^\opt) \right|^q \right] \right)^{\frac{1}{q}-1} \left( \nabla_\pi^2 V(0, \pi^\opt) \right)^{-1}
 \cdot \E_{\P}\left[ \frac{\nabla_{x}\nabla_\pi f(X,\pi^\opt) \nabla_x f(X,\pi^\opt)} {|\nabla_x f(X, \pi^\opt)|^{2-q}}\right]
\end{align*}
for $f(x,\pi)=u(\langle x, \pi\rangle)$.
Writing $u^\prime := u^\prime( \langle x,\pi^\opt \rangle )$ and $u^{\prime\prime} := u^{\prime\prime}( \langle x,\pi^\opt \rangle )$ in order to simplify notation, an explicit computation yields
\begin{align*}
    \nabla_x f(x,\pi^\ast) & = \pi^\opt u^\prime,\\
    \nabla_\pi\nabla_{x} f(x,\pi^\opt) & =  \pi^\opt x^\top u^{\prime\prime}  + Iu^\prime ,\\
   \nabla_\pi \nabla_{x} f(x,\pi^\opt) \nabla_x f(x,\pi^\ast) & = \left(\langle x,\pi^\opt \rangle u^\prime u^{\prime\prime}   + \left( u^\prime \right)^2 \right)\cdot \pi^\opt,\\
    \nabla_{\pi}^2 f(x,\pi^\opt) & = x x^\top u^{\prime \prime},\\
   \nabla_\pi^2 V(0) & = \E_{\P}\left[ XX^\top u^{\prime \prime}\right],
\end{align*}
where $I\in \R^{d \times d}$ denotes the identity matrix and by Assumption \ref{Ass:2} the matrix $$\E_{\P}\left[ XX^\top u^{\prime \prime} \right]$$ is negative definite. 
In particular we have
\begin{align*}
    \pip(0) &= \left( \E_{\P} \left[ \left| \nabla_x f(X,\pi^\opt) \right|^q \right] \right)^{\frac{1}{q}-1} \left( \nabla_\pi^2 V(0, \pi^\opt) \right)^{-1}\cdot\E_{\P} \left[\frac{\nabla_{x}\nabla_\pi f(X,\pi^\opt) \nabla_x f(X,\pi^\opt)} {|\nabla_x f(X, \pi^\opt)|^{2-q}}\, \right]\\
    &= \left( \E_{\P}\left[ \left| u^\prime \right|^q \right] \right)^{\frac{1}{q}-1} |\pi^\opt|^{1-q} \left( \nabla_\pi^2 V(0, \pi^\opt) \right)^{-1}\E_{\P}\left[ \frac{\langle X, \pi^\opt \rangle u^\prime u^{\prime\prime} +(u^\prime)^2}{|\pi^\opt u^\prime |^{2-q}}\, \right]\cdot \pi^\opt\\
   &= \|u^\prime\|_{L^q(\P)}^{1-q}
    \left(\E_{\P}\left[  \frac{ \langle X,\pi^\opt \rangle u^{\prime\prime} +u^\prime }{\left| u^\prime \right|^{1-q} }\right]  \right) \left(\nabla_\pi^2 V(0) \right)^{-1} \cdot \frac{\pi^\opt}{\left| \pi^\opt \right|},
\end{align*}
which concludes the proof for the case $p\in (1, \infty)$.\\

\textit{Step 2:} The case $p=\infty$ follows again by going through \cite[Proof of Theorem 4]{bartl2020robust} line by line, but some arguments can be cut short in this case: indeed, let us first note that there is only one measure in the ball $B_\delta(\P)$ attaining the value $$\inf_{\tilde{\P}\in B_\delta(\P)} \E_{\tilde{\P}} \left[ u( \langle X, \pi\rangle )\right],$$
namely $\P^\opt_\delta=(x-\delta\text{sign}(\pi^\opt))_\#\P$. On the other hand, by optimality of $\pi^\opt_\delta$ we have
\begin{align*}
\E_{\P^\opt_\delta} \left[ X u^\prime (\langle X, \pi^\opt_\delta)\rangle\right] =0,
\end{align*}
and under the uniform integrability condition stated in Assumption \ref{Ass:2}(i)(b) one can directly compute
\begin{align*}
-\lim_{\delta\to 0}\frac{\E_{\P} \left[ X u^\prime (\langle X, \pi^\opt_\delta)\rangle\right]}{\delta}
&=\lim_{\delta \to 0}\frac{\E_{\P^\opt_\delta} \left[ X u^\prime (\langle X, \pi^\opt_\delta)\rangle\right]-\E_{\P} \left[ X u^\prime (\langle X, \pi^\opt_\delta)\rangle\right]}{\delta}\\
&=\lim_{\delta \to 0}\Bigg(\frac{\E_{\P} \left[ (X-\delta\text{sign}(\pi^\opt)) u^\prime (\langle X-\delta\text{sign}(\pi^\opt), \pi^\opt_\delta)\rangle)\right]}{\delta}\\
&\qquad\qquad- \frac{\E_{\P}\left[X u^\prime (\langle X, \pi^\opt_\delta\rangle)\right]}{\delta}\Bigg)\\
&=-\E_{\P} \left[\langle X, \pi^\opt \rangle u^{\prime\prime}(\langle X, \pi^\opt \rangle)+u^\prime (\langle X, \pi^\opt \rangle) \right]\cdot \text{sign}(\pi^\opt).
\end{align*}
Lastly, as the matrix $\nabla_\pi^2 V(0)$ is invertible, we can again follow the same arguments as in \cite[Proof of Theorem 4]{bartl2020robust} to obtain
\begin{align*}
\pip(0)&= \left(\nabla_\pi^2 V(0)\right)^{-1} \lim_{\delta\to 0}\frac{\E_{\P} \left[ X u^\prime (\langle X, \pi^\opt_\delta \rangle)\right]}{\delta}\\
&= \left(\nabla_\pi^2 V(0)\right)^{-1} \E_{\P} \Big[\langle X, \pi^\opt \rangle u^{\prime\prime}(\langle X, \pi^\opt \rangle)+u^\prime (\langle X, \pi^\opt \rangle) \Big]\cdot \text{sign}(\pi^\opt).
\end{align*}
This shows the claim for the case $p=\infty$.
\end{proof}

\section{Distributionally robust marginal utility pricing}\label{sec:robust Davis price}

We move now to the discussion and rigorous statements of results announced in section \ref{sec:mainresults_marginalprice} above.

\subsection{Existence and regularity}
Before we state the main theorem of this section, we set up some useful notation and state some immediate consequences of our setup. For this, we will first keep $\d>0$ fixed. We recall $V(\delta,\epsilon,\d)$ defined in \eqref{eq:Vpd def}.

\begin{definition}
 Let us define 
\begin{equation}\label{eq:v:u of eps g}
 v(\epsilon, x,\pi)=u\left(-\epsilon+\langle x,\pi\rangle+\frac{\epsilon}{\d} g(x)\right)
\end{equation}
 and  $$V(\delta,\epsilon, \pi) \coloneq \inf_{\tilde{\P}\in B_{\delta}(\P)} \E_{\tilde{\P}} \left[ v(\epsilon, X,\pi)\,\right].$$  Given a strategy $\pi\in \Adm$, write $$B^{\opt,\epsilon}_{\delta}(\P,\pi) := \left\{ \P^\opt \in B_{\delta}(\P)\text{ such that } \E_{\P^\opt} \left[ v(\epsilon,X,\pi)\right]=V(\delta,\epsilon,\pi)  \right\}$$ for the set of minimising measures. Lastly, denote the set of optimising vectors by $$\Adm^{\opt,\epsilon}_{\delta}:= \left\{\pi\in \Adm\ : \ V(\delta,\epsilon,\pi)=V(\delta,\epsilon,\d) \right\}.$$ 
We denote a generic element of  $\Adm^{\opt,\epsilon}_{\delta}$ by $\pi^\opt(\delta,\epsilon)$ and recall that $\pi^\opt=\pi^\opt(0,0)$ as well as $\pi^\opt_\delta=\pi^\opt(\delta,0)$.
\end{definition}
In particular we note that the sets $\Adm^{\opt, 0}_{\delta}$ and $B^{\opt,0}_{\delta}(\P,\pi)$ are independent of $\d$ and the payoff $g$. We now make the following assumption:
\begin{assumption}\label{Ass:3}
There exists $k\in (0,\infty)$ such that the function $v:[-k,k] \times \mathcal{S} \times \Adm \to \R$ in \eqref{eq:v:u of eps g} is continuously differentiable. Furthermore,
\begin{enumerate}[(i)]
\item If $p\in (1, \infty)$ then for each $r>0$ there exist $c,\gamma>0$ such that 
\begin{enumerate}
\item $|v(\epsilon,x,\pi)|\le c(1+|x|^{p-\gamma})$,
\item $|\nabla_{\epsilon}v(\epsilon,x,\pi)|\le c(1+|x|^{p-\gamma})$
\end{enumerate}
for all $x\in \mathcal{S}, \pi\in \Adm$ such that $|\pi|\le r$ and $\epsilon\in [-k,k]$.
\item If $p=\infty$, then there exists some $\tilde{\delta}>0$ such that for each $r>0$
\begin{align*}
\E_{\P} \left[\sup_{\pi\in \Adm, |\pi|\le r}\sup_{|z|\le \tilde{\delta}}\sup_{\epsilon\in [-k,k]} \left|\nabla_\epsilon v( \epsilon, X+z, \pi)\right|\right]<\infty.
\end{align*}
\end{enumerate}
\end{assumption}

\begin{remark}
Recall that $u:\R\to \R$ is assumed to be continuously differentiable. Thus
\begin{align*}
\nabla_{\epsilon}v(\epsilon,x,\pi)=u^\prime \left(-\epsilon+\langle x,\pi\rangle+\frac{\epsilon}{\d} g(x)\right)\left(-1 +\frac{g(x)}{\d}\right).
\end{align*}
The assumption that $v$ is continuously differentiable thus implies in particular that $g:\mathcal{S} \to \R$ is continuous.
\end{remark}

\begin{remark}
Note that Assumption \ref{Ass:3} is the natural robust counterpart of Assumption \ref{Ass:5} in appendix \ref{app:marg_util} and is thus stronger than Assumption \ref{Ass:5}. We also remark that when setting $\epsilon=0$ in item (i) of Assumption \ref{Ass:3} we obtain in particular
\begin{align*}
 |u(\langle x, \pi\rangle )|&\le c(1+|x|^{p-\gamma}),\\
 \left|u^\prime (\langle x, \pi\rangle) \left(-1+\frac{g(x)}{p_d}\right)\right|&\le c(1+|x|^{p-\gamma}),
\end{align*}
which is usually stronger than Assumption \ref{Ass:1}.
\end{remark}

Before we consider a characterisation of the distributionally robust Davis price in the spirit of \eqref{eq:davis_main} we first establish existence and uniqueness of optimisers $\pi^\opt(\delta,0)\in \Adm^{\opt,0}_{\delta}$ and $\P^{\opt}\in B^{\opt,0}_{\delta}(\P,\pi(\delta,0)).$

\begin{lemma}\label{lem:convergence}
Let Assumption \ref{Ass:3} be satisfied. Then the following hold: 
\begin{enumerate}[(i)]
\item The set $\Adm^{\opt,\epsilon}_{\delta}\neq \emptyset$ for all $\delta\ge 0$ and there exists a compact set $K\subseteq \R^d$ such that the set of optimisers is contained in $K$, i.e. $$\bigcup_{\epsilon\in[-k,k]} \bigcup_{\delta\ge 0}\Adm^{\opt,\epsilon}_{\delta}\subseteq K.$$
\item For every sequence $(\epsilon_n)_{n\in \N}$ such that $\lim_{n\to \infty} \epsilon_n=0$, fixed $\delta\ge 0$ and $(\pi_n)_{n\in \N}$ such that $\pi_n\in \Adm^{\opt,\epsilon_n}_{\delta}$ for all $n\in \N$ there exists a subsequence which converges to some $\pi^\ast\in \Adm^{\opt,0}_{\delta}$.
\item For every sequence $(\delta_n)_{n\in \N}$ such that $\lim_{n\to \infty} \delta_n=0$ and $(\pi_n)_{n\in \N}$ such that $\pi_n\in \Adm^{\opt,0}_{\delta_n}$ for all $n\in \N$ there exists a subsequence which converges to $\pi^\ast\in \Adm^{\opt,0}_{0}$.
\end{enumerate}
\end{lemma}

\begin{proof}
The first claim follows exactly as in the proof of \cite[Lemma 19]{bartl2020robust} exchanging the ``$\inf$" and the ``$\sup$" and
noting that
\begin{align*}
\inf_{\tilde{\P} \in B_{\delta}(\P)} \E_{\tilde{\P}}\left[ v(\epsilon, X, \pi)\right] \le \E_{\P}\left[ v(\epsilon, X,\pi)\right]
\end{align*}
for all $\delta\ge 0, \pi \in \Adm$ and $\epsilon\in [-k,k]$, see also the proof of Lemma \ref{lem:strict_concave}, noting that $V(\delta,\epsilon)> -\infty$ by Assumption \ref{Ass:3}.(i). For \textit{(ii)} we argue by contradiction: let us  assume that (possibly after taking a subsequence) $\pi_n$ converges to a limit $\tilde{\pi}\notin \Adm^{\opt, 0}_{\delta}$. Since $\tilde{\pi}$ is not an optimiser, and using the reverse Fatou lemma, we have
\begin{align*}
V(\delta, 0, \d)> & \inf_{\tilde{\P} \in B_{\delta}(\P)} \E_{\tilde{\P}} \left[ v(0, X,\tilde{\pi})\right] \ge \limsup_{n\to \infty}  \inf_{\tilde{\P} \in B_{\delta}(\P)} \E_{\tilde{\P}} \left[ v(\epsilon_n,  X, \pi_n)\right] \\
&=  \limsup_{n\to \infty} V(\delta, \epsilon_n, \d).
\end{align*}
On the other hand, plugging $\pi^\opt\in \Adm^{\opt, 0}_\delta$ into $V(\delta,\epsilon_n,\d)$ implies
\begin{align*}
\liminf_{n \to \infty} V(\delta,\epsilon_n, \d) &\ge \liminf_{n \to \infty} \inf_{\tilde{\P} \in B_{\delta}(\P)} \E_{\tilde{\P}} \left[ v(\epsilon_n, X, \pi^\opt)\right] \\
&\ge \liminf_{n \to \infty} \Bigg( \inf_{\tilde{\P} \in B_{\delta}(\P)} \E_{\tilde{\P}} \left[ v(0, X, \pi^\opt)\right] \\
&\qquad-\epsilon_n \sup_{\tilde{\P} \in B_{\delta}(\P)} \E_{\tilde{\P}} \left[  \int_0^1 \nabla_{\epsilon} v(t\epsilon_n, X, \pi^\opt)\,dt\right]  \Bigg)\\
&=V(\delta,0, \d)
\end{align*}
as $|\nabla_{\epsilon} v(\epsilon, x, \pi^\opt)|\le c(1+|x|^{p})$ for some $c>0$, all $\epsilon$ small enough and all $x\in \mathcal{S}$. This gives the desired contradiction concludes the proof of the second part. The proof of \textit{(iii)} follows again as in the proof of \cite[Lemma 19]{bartl2020robust} exchanging the ``$\inf$" and the ``$\sup$". This concludes the proof.
\end{proof}

While existence of optimisers is thus guaranteed, their uniqueness is more delicate. Let us start with the specific case $\pi^\opt=0$. Interestingly, it turns out that under certain assumptions, this case fully characterises the martingale property of $\P$ (recall that  $X=S_1-S_0)$:
\begin{lemma}\label{lem:degenerate}
Assume that $0\in \interior{\Adm}$. Then $\sup_{\pi\in \Adm} \E_{\P}\left[ u(\langle X,\pi \rangle)\right]$ is attained for $\pi^\opt=0$ if and only if $\E_{\P}\left[ X\right]=0$. In that case, for all $\delta\ge 0$,
\begin{align*}
\sup_{\pi\in \Adm} \inf_{\tilde{\P}\in B_{\delta}(\P)} \E_{\tilde{\P}} \left[ u(\langle X,\pi \rangle)\right]
\end{align*}
is attained for $\pi^\opt(\delta,0)=0$. Furthermore, the optimiser is unique, $\Adm^{\opt, 0}_\delta=\{0\}$, and $B^{\opt, 0}_{\delta}(\P, \pi^\opt(\delta,0))=B_\delta(\P)$.
\end{lemma}

\begin{proof}
Note that the first-order condition for $\delta=0$ at $\pi^\opt=0$ implies
\begin{align*}
0=\E_{\P}\left[ X u^\prime(\langle X, \pi^\opt \rangle)\right] =u^\prime(0)\,\E_{\P} \left[ X\right],
\end{align*}
and thus $\E_{\P} \left[ X\right]=0$, as $u$ is strictly increasing.
On the other hand, if $\E_{\P} \left[ X\right]=0$, Jensen's inequality implies
\begin{align*}
\sup_{\pi\in \Adm} \E_{\P} \left[ u(\langle X,\pi \rangle)\right] \le \sup_{\pi\in \Adm} u\left( \E_{\P} \left[ \langle X, \pi \rangle \right] \right)=u(0),
\end{align*}
and the supremum is attained for $\pi^\opt=0$. 
Furthermore, if $\pi^\opt=0$ then
\begin{align}\label{eq:easier}
\begin{split}
u(0)&= \inf_{\tilde{\P}\in B_{\delta}(\P)} \E_{\tilde{\P}} \left[  u\left(\langle X, 0\rangle \right)\right] \le \sup_{\pi\in \Adm} \inf_{\tilde{\P}\in B_{\delta}(\P)} \E_{\tilde{\P}} \left[ u(\langle X,\pi \rangle)\right] \\
&\le \sup_{\pi\in \Adm} \E_{\P}\left[ u(\langle X,\pi \rangle)\right] =u(0),
\end{split}
\end{align}
so in fact equality holds in \eqref{eq:easier} and the supremum is again attained for $\pi^\opt(\delta,0)=0$.
Lastly, we note that for any $\pi\in \Adm, \pi\neq 0$ we have $$\P(\{X\in \mathcal{S}\colon \langle X,\pi\rangle>0 \})>0$$ as stated in Assumption \ref{Ass:4a}. Thus it is always possible to find $\hat{\P} \in B_{\delta}(\P)$ such that $\E_{\hat{\P}} \left[ \langle X, \pi\rangle \right]<0$, and so
\begin{align*}
 \inf_{\tilde{\P}\in B_{\delta}(\P)} \E_{\tilde{\P}} \left[ u(\langle X,\pi \rangle)\right] &\le  \E_{\hat{\P}} \left[ u(\langle X,\pi \rangle)\right] \\
 &\le u\left( \E_{\hat{\P}}\left[ \langle X,\pi \rangle\right] \right) <u(0),
\end{align*}
where we used Jensen's inequality again.
In conclusion $\pi^\opt_\delta=0$ is the unique optimiser. 
This concludes the proof.
\end{proof}

More generally, the lemma below shows that $\pi^\opt_\delta \in \mathcal{A}^{\opt, 0}_\delta$ is always unique. In the case $0\notin \Adm^\opt$ we also obtain uniqueness of $\P^\opt\in B_\delta^{\opt,0}(\P, \pi^\opt_\delta)$.

\begin{lemma}
Let Assumption \ref{Ass:3} hold. For $\delta>0$ small enough the optimal strategy is unique, $\Adm^{\opt, 0}_\delta=\{\pi^\opt_\delta\}$. Furthermore $\pi^\opt_\delta \to \pi^\opt$ for $\delta \to 0$, where $\pi^\opt$ is again unique. If in addition $0\notin \mathcal{A}^\opt$, then $B_\delta^{\opt, 0}(\P, \pi^\opt)=\{\P^\opt\}$ is a singleton. 
\end{lemma}

\begin{proof}
Take any $\pi, \tilde{\pi}\in \Adm$ such that $\pi\neq \tilde{\pi}$. Assumption \ref{Ass:3} and \cite[Lemma 20]{bartl2020robust} implies that the set $B_{\delta}^{\opt,0}(\P,\lambda \pi+(1-\lambda)\tilde\pi)\neq \emptyset$ for all $\lambda\in [0,1]$. Thus for all $\lambda\in (0,1)$ we have for any $\tilde{\P} \in B_{\delta}^{\opt,0}(\P,\lambda\pi+(1-\lambda)\tilde{\pi})$ that 
\begin{align*}
V(\delta, 0, \lambda \pi+(1-\lambda)\tilde{\pi})&=\E_{\tilde{\P}} \left[  v(0,x,\lambda \pi+(1-\lambda)\tilde{\pi})\right] \\
&>\lambda\E_{\tilde{\P}}\left[ v(0,X,\pi)\right]+(1-\lambda)\E_{\tilde{\P}} \left[ v(0,X,\tilde{\pi})\right]\\
&\ge \lambda V(\delta, 0,\pi)+(1-\lambda) V(\delta, 0,\tilde{\pi})
\end{align*}
for all $\delta\ge 0$ small enough using Assumption \ref{Ass:4a} and the fact $u$ is strictly concave. In conclusion $\pi \mapsto V(\delta, 0, \pi)$ is strictly concave on $\Adm$ and thus the optimiser $\pi^\opt\in \mathcal{A}^{\opt, 0}_\delta$ is unique. Next, Lemma \ref{lem:convergence} implies that $\pi^\opt_\delta \to \pi^\opt$ for $\delta \to 0$, where $\pi^\opt$ is again unique by strict concavity of $u$.\\

On the other hand, for $\pi^\opt_\delta \neq 0$ and $1 < p <\infty$, uniqueness of $\P^\opt \in B^{\opt,0}_\delta(\P,\pi^\opt_\delta)$ follows from strict convexity of $L^p$-spaces: indeed if there exists $\P'\in B^{\opt,0}_\delta(\P,\pi^\opt_\delta)$, then $$\frac{1}{2} \left( \P^\opt+\P'\right)\in \interior{(B_\delta(\P))} $$ and $$ \frac{1}{2} \left( \E_{\P^\opt}[v(0,X,\pi^\opt_\delta)] + \E_{\P'} [v(0,X,\pi^\opt_\delta)] \right)= V(\delta,0, \pi^\opt_\delta)=V(\delta).$$
However, as $0\notin \mathcal{A}^\opt$, we conclude that $V(\delta)<V(0)$. This leads to a contradiction.\\
Lastly, for $p=\infty$ and $\pi^{\opt}_\delta \neq 0$, it can be directly seen that $\P^\opt=(x-\text{sign}(\pi^\opt_\delta))_\#\P$ is unique as well.
\end{proof}

We now compute the derivative $\nabla_\epsilon V(\delta, 0 ,\d)$ for fixed $\delta\ge 0$ and $\d> 0$.
\begin{theorem}\label{thm. marginal}
Fix $\delta\ge 0, \d > 0$ and let Assumption \ref{Ass:3} hold. Then
\begin{align*}
\begin{split}
\nabla_{\epsilon} V(\delta, 0, \d)&= \sup_{\pi \in \Adm^{\opt, 0}_{\delta} } \inf_{\P^\opt \in B^{\opt, 0}_{\delta}(\P, \pi) } \E_{\P^\opt} \left[ \nabla_{\epsilon} v(0, X, \pi)\right] \\
&=\sup_{\pi \in \Adm^{\opt, 0}_{\delta} } \inf_{\P^\opt \in B^{\opt, 0}_{\delta}(\P, \pi) }  \E_{\P^\opt} \left[ u^{\prime}( \langle X,\pi\rangle) \left(-1+\frac{g(X)}{\d}\right)\right].
\end{split}
\end{align*}
In particular the robust marginal utility price $\hatd(\delta)$ is a solution to 
\begin{align}\label{eq:marginal}
\inf_{\P^\opt \in B^{\opt, 0}_{\delta}(\P, \pi^\opt_\delta) }  \E_{\P^\opt} \left[ u^{\prime}( \langle X,\pi^\opt_\delta\rangle) \left(-1+\frac{g(X)}{\d}\right)\,\right]=0
\end{align}
and is thus given by 
\begin{align*}
\hatd(\delta)=\frac{ \E_{\P^\opt}\left[ u^\prime (\langle X,\pi^\opt_\delta\rangle)\,g(X)\, \right]}{\E_{\P^\opt} \left[ u^\prime(\langle X, \pi^\opt_\delta\rangle)\right]}
\end{align*}
for any $\P^\opt \in B_\delta^{\opt,0}(\P, \pi^\opt(\delta,0))$.
\end{theorem}

\begin{proof}[Proof of Theorem \ref{thm. marginal}]
\textit{Step 1:} We first consider the case $p\in (1, \infty).$ We start by proving the inequality
\begin{align}\label{eq:firstact}
\lim_{\varepsilon \to 0} \frac{V(\delta,\epsilon, \d)-V(\delta, 0, \d)}{\epsilon}\le \sup_{\pi \in \Adm^{\opt, 0}_{\delta} } \inf_{\P^\opt \in B^{\opt, 0}_{\delta}(\P, \pi) } \E_{\P^\opt}\left[ \nabla_{\epsilon} v(0, X, \pi)\right].
\end{align}
Indeed, for small $\epsilon>0$ let us take optimisers $\pi^{\opt}(\delta, \epsilon)\in \Adm^{\opt,\epsilon}_{\delta}$. Then, we can apply Lemma \ref{lem:convergence}.\textit{(ii)}, so that after taking a subsequence (without relabelling) there exists some $\pi\in \Adm^{\opt, 0}_{\delta}$ such that $\lim_{\epsilon\to 0}\pi^{\opt}(\delta, \epsilon)=\pi$. Take $\P^\opt \in B^{\opt,0}_{\delta}(\P,\pi)$. As $\P^\opt \in B_{\delta}(\P)$ we then have 
\begin{align*}
\lim_{\varepsilon \downarrow 0} \frac{V(\delta,\epsilon, \d)-V(\delta,0, \d)}{\epsilon}&\le \lim_{\epsilon \to 0}\frac{1}{\epsilon} \E_{\P^\opt} \left[ v(\epsilon,X,\pi^{\opt}(\delta, \epsilon))-v(0,X, \pi)\right].
\end{align*}
By assumption we have $|\nabla_{\epsilon}v(\tilde{\epsilon}, x, \pi)| \le c (1+|x|^{p-\gamma})$. Then arguing as in \eqref{eq:abscont}, the envelope theorem for arbitrary choice sets of \cite[Theorem 3]{milgrom2002envelope} implies that
\begin{align*}
\lim_{\varepsilon \to 0} \frac{V(\delta,\epsilon, \d)-V(\delta, 0, \d)}{\epsilon}\le  \E_{\P^\opt}\left[ \nabla_{\epsilon} v(0,X,\pi)\right].
\end{align*}
As $\P^\opt \in B^{\opt,0}_{\delta}(\P,\pi)$ was arbitrary, we have shown \eqref{eq:firstact}.

We proceed to show that 
\begin{align}\label{eq:seccondact}
\lim_{\varepsilon \to 0}  \frac{V(\delta,\epsilon, \d)-V(\delta,0, \d)}{\epsilon}\ge  \sup_{\pi \in \Adm^{\opt, 0}_{\delta} } \inf_{\P^\opt \in B^{\opt, 0}_{\delta}(\P, \pi) } \E_{\P^\opt}  \left[ \nabla_{\epsilon} v(0, X, \pi)\right].
\end{align}

Take $\pi \in \Adm^{\opt,0}_{\delta}$.
For every sufficiently small $\epsilon>0$, let $\P^{\opt,\epsilon}\in B^{\opt, \epsilon}_\delta(\P,\pi)$, so that $V(\delta, \epsilon,\pi^{\opt})=\E_{\P^{\opt,\epsilon}}\left[ v(\epsilon, X,\pi^{\opt})\right]$. The existence of such $\P^{\opt,\epsilon}$ is guaranteed  by the assumption $| v(\epsilon,x,\pi)|\le c(1+|x|^{p-\gamma})$ for all $x\in \mathcal{S}$ and \cite[Lemma 20]{bartl2020robust}, which also guarantees that (possibly after passing to a subsequence) there is $\P^\opt \in B_\delta(\P)$ such that $\P^{\opt,\epsilon} \to\P^\opt$ in $W_{p-\gamma}$. 
We claim that $\P^\opt \in B_\delta^{\opt, 0}(\P,\pi)$.
	Indeed, as $|\nabla_{\epsilon}v(\tilde{\epsilon},x,\pi)|\le c(1+|x|^{p-\gamma})$ and $|v(0,x,\pi)|\le c(1+|x|^{p-\gamma})$ one has
	\begin{align*}
	\lim_{\epsilon\to 0} V(\delta,\epsilon,\pi)&=\lim_{\epsilon\to 0} \E_{\P^{\opt,\epsilon}} \left[ v(\epsilon,X,\pi)\right]\\
	& =\lim_{\epsilon\to 0} \left( \E_{\P^{\opt,\epsilon}}\left[  v(0,X,\pi)\right]+\epsilon\,\E_{\P^{\opt,\epsilon}} \left[ \int_{0}^1  \nabla_{\epsilon}v(t \epsilon,X,\pi)\,dt \right] \right)\\
	& =\E_{\P^\opt} \left[ v(0,X,\pi)\right]
	\ge V(\delta, 0,\pi).
		\end{align*}
	On the other hand, for any choice $\tilde{\P}\in B^{\opt, 0}_\delta(\P,\pi)$ one has by Assumption \ref{Ass:3} and dominated convergence
	\[\lim_{\epsilon \to 0} V(\delta, \epsilon,\pi)
	\le \lim_{\epsilon\to 0} \E_{\tilde{\P}} \left[  v(\epsilon,X,\pi)\right]
	= \E_{\tilde{\P}}\left[  v(0,X,\pi)\right]
	=V(\delta,0,\pi).\]
	This implies $V(\delta, 0,\pi)=\E_{\P^\opt} \left[ v(0,X,\pi)\right]$ and in particular $\P^\opt \in B_\delta^{\opt, 0}(\P,\pi)$. At this point we expand $v(\epsilon,x,\pi)=v(0,x,\pi)+\epsilon \int_0^1 \nabla_\epsilon v(t\epsilon, x, \pi)\, dt$ so that 
	\begin{align*}  
    &V(\delta, \epsilon,\pi^{\opt} (\delta, \epsilon))-V(\delta, 0,\pi)\ge V(\delta,\epsilon,\pi)-V(\delta,0,\pi) \\
    &= \E_{\P^{\opt,\epsilon}} \left[ \Big( v(0,X,\pi)+\epsilon \int_0^1  \nabla_\epsilon v(t\epsilon, X,\pi)\,dt\Big)\right]- \E_{\P^\opt}\left[ v(0,X,\pi)\right]\\
    &\ge \E_{\P^{\opt,\epsilon}}\left[ \Big( v(0,X,\pi)+\epsilon \int_0^1  \nabla_\epsilon v(t\epsilon, X,\pi)\,dt\Big)\right]- \E_{\P^{\opt,\epsilon}}\left[ v(0,X,\pi)\right]\\
    &= \epsilon \E_{\P^{\opt,\epsilon}}\left[ \int_0^1  \nabla_\epsilon v(t\epsilon,X,\pi)\,dt\,\right],
    \end{align*}
    where we used $\P^\opt \in B_\delta^{\opt,0}(\P,\pi)$ for the first equality.
    Fix $t\in [0,1]$ and recall that $|\nabla_{\epsilon}v(\tilde{\epsilon}, x, \pi)| \le c (1+|x|^{p-\gamma})$ for all $x \in \mathcal{S}$, $\pi\in K$, $\tilde{\epsilon}>0$ small enough and $(\tilde{\epsilon},x) \mapsto \nabla_{\epsilon} v(\tilde{\epsilon}, x, \pi)$ is continuous. As $t\epsilon\to 0$ and $\P^{\opt,\epsilon}$ converge to $\P^\opt$ in $W_{p-\gamma}$ for $\epsilon \to 0$, we conclude that the measures $ \delta_{t\epsilon}\otimes \P^{\opt,\epsilon}$ converge (on the product space) in $W_{p-\gamma}$ to the measure $\delta_{0}\otimes\P^\opt$ (here $\otimes$ denotes the product measure). As  $$\E_{\P^{\opt,\epsilon}} \left[\nabla_{\epsilon}(t\epsilon,X,\pi)\right]=\E_{\delta_{t\epsilon}\otimes \P^{\opt,\epsilon}}\left[ \nabla_{\epsilon}(\tilde{\epsilon},X,\pi)\right]$$ and similarly 
    $$\E_{\P^{\opt}} \left[\nabla_{\epsilon}(0,X,\pi)\right]=\E_{\delta_{0}\otimes \P^{\opt}}\left[ \nabla_{\epsilon}(\tilde{\epsilon},X,\pi)\right],$$
    we conclude that 
    $$ \E_{\P^{\opt,\epsilon}}\left[ \nabla_{\epsilon}(t\epsilon,X,\pi)\right] \to \E_{\P^\opt}\left[ \nabla_{\epsilon}(0,X,\pi)\right].$$
    Using the above together with Fubini's theorem and the dominated convergence theorem, we finally conclude that
    \begin{align*}
    \lim_{\epsilon\to 0}\frac{1}{\epsilon} \epsilon \E_{\P^{\opt,\epsilon}}\left[ \int_0^1  \nabla_\epsilon v(t\epsilon,X,\pi)\,dt\right]
    &= \lim_{\epsilon\to 0} \int_0^1  \E_{\P^{\opt,\epsilon}} \left[\nabla_\epsilon v(t\epsilon,X,\pi)\right]\,dt \\
    &=\int_0^1 \E_{\P^\opt} \left[ \nabla_\epsilon v(0,X,\pi)\right]\,dt\\
    &=  \E_{\P^\opt} \left[ \nabla_\epsilon v(0,X,\pi)\right],
    \end{align*}
    which ultimately shows \eqref{eq:seccondact} for $p\in (1, \infty)$. \\

\textit{Step 2:} Let us now consider the case $p=\infty$. Inequality \eqref{eq:firstact} follows as before, noting that the condition 
    \begin{align}\label{eq:ass_inf}
\E_{\P} \left[\sup_{\pi\in \Adm, |\pi|\le r}\sup_{|a|\le s}\sup_{\epsilon\in [-k,k]} \left|\nabla_\epsilon v( \epsilon, X+a, \pi)\right|\right]<\infty
\end{align}
implies by the dominated convergence theorem that
\begin{align*}
\nabla_\epsilon \E_{\P^\opt}\left[  v(\epsilon, X, \pi)\right]=\E_{\P^\opt}\left[ \nabla_{\epsilon} v(\epsilon, X, \pi)\right]
\end{align*}
for all $\epsilon\in [-k,k]$. Now we can apply again the envelope theorem for arbitrary choice sets of \cite[Theorem 3]{milgrom2002envelope} to conclude that
\begin{align*}
\lim_{\varepsilon \to 0} \frac{V(\delta,\epsilon, \d)-V(\delta, 0, \d)}{\epsilon}\le  \E_{\P^\opt}\left[ \nabla_{\epsilon} v(0,X,\pi)\right].
\end{align*}
 We now show inequality \eqref{eq:seccondact}. Note that in this case the existence of $\P^{\opt,\epsilon}\in B^{\opt, \epsilon}_\delta(\P,\pi)$, so that $V(\delta, \epsilon,\pi^{\opt})=\E_{\P^{\opt,\epsilon}}\left[ v(\epsilon, X,\pi^{\opt})\right]$  is guaranteed by Prokhorov's theorem, the Portmanteau lemma, \eqref{eq:ass_inf} and the dominated convergence theorem. In the same way it can be shown that (possibly after passing to a subsequence) there is $\P^\opt \in B_\delta(\P)$ such that 
\begin{align*}
\lim_{\epsilon\to 0} \E_{\P^{\opt,\epsilon}}\left[ v(\epsilon, X, \pi)\right]=\E_{\P^\opt}\left[ v(0, X, \pi)\right].
\end{align*}
In particular the same arguments as in \emph{Step 1} above imply $V(0,\pi)=\E_{\P^\opt} \left[ v(0,X,\pi)\right]$. Lastly we conclude using Fubini's theorem, \eqref{eq:ass_inf} and the dominated convergence theorem as above that 
\begin{align*}
\lim_{\epsilon\to} \frac{V(\delta, \epsilon,\pi^{\opt} (\delta, \epsilon))-V(\delta, 0,\pi)}{\epsilon} 
&\ge \lim_{\epsilon\to 0}\E_{\P^{\opt,\epsilon}}\left[ \int_0^1  \nabla_\epsilon v(t\epsilon,X,\pi)\,dt\,\right]\\
&=\E_{\P^{\opt}}\left[ \nabla_\epsilon v(0,X,\pi)\right].
\end{align*}
This concludes the proof.
\end{proof}

In conclusion, Theorem \ref{thm. marginal} states that for any pair of optimisers $(\pi^\opt_\delta, \P^\opt)\in  \Adm^{\opt,0}_{\delta}\times B^{\opt, 0}_{\delta}(\P, \pi) $ -- and we recall that apart from the extreme case $\pi^\opt=0$ this pair is in fact unique -- the robust marginal utility price $\hatd(\delta)$ can be written as
\begin{align*}
\hatd(\delta)=\frac{\E_{\P^\opt} \left[ u^\prime ( \langle X, \pi^\opt_\delta\rangle) \,g(X)\right]}{\E_{\P^\opt} \left[ u(\langle X, \pi^\opt_\delta \rangle)\right]}.
\end{align*}
We remark that the same arguments as in the classical case given in Section \ref{sec:marginal_util} imply that $\hatd(\delta)$ is the expectation under a martingale measure equivalent to $\P^\opt$, in particular the optimiser $\P^\opt$ does not allow for arbitrage and is close in Wasserstein sense to $\P$. 

\subsection{Sensitivity of distributionally robust marginal utility price}

We already know from results in \cite{bartl2020robust}, that for small $\delta>0$, the maximising measure $\P^\opt\in B_\delta(\P)$ is built from $\P$ as a push forward in the direction $-\pi^\opt u^\prime(\langle X, \pi^\opt\rangle)$, which is $\pi^\opt$ times the Radon-Nikodym derivative $d\Q_u/d\P$ of the agent's subjective martingale measure $\Q_u$ defined in \eqref{eq:u mart mg}. For small $\delta>0$ we thus have that
\begin{align*}
\hatd(\delta)\approx &\frac{1}{\E_{\P}\left[ u^\prime(\langle X-\delta \pi^\opt u^\prime(\langle X, \pi^\opt \rangle), \pi^\opt+\delta \pip(0) \rangle)\right] }\\
&\quad \cdot \E_{\P} \Big[ u^\prime ( \langle X-\delta \pi^\opt u^\prime(\langle X, \pi^\opt \rangle), \pi^\opt+\delta \pip(0)\rangle)\\
&\qquad\qquad \cdot g(X-\delta \pi^\opt u^\prime(\langle X, \pi^\opt \rangle))\Big]\\
&=\E_{\Q_u^\delta} \left[  g(X-\delta \pi^\opt u^\prime(\langle X, \pi^\opt \rangle))\right],
\end{align*}
where 
\begin{align*}
\frac{d\Q_u^\delta}{d\P}(x)=\frac{u^\prime ( \langle x-\delta \pi^\opt u^\prime(\langle x, \pi^\opt \rangle), \pi^\opt+\delta  \pi'(0)\rangle)}{\E_{\P}\left[ u^\prime (\langle X-\delta \pi^\opt  u^\prime(\langle X, \pi^\opt \rangle), \pi^\opt+\delta  \pip (0) \rangle)\right] },
\end{align*}
where we recall from \eqref{eq:sensitivity_optimiser} that
\begin{align*}
\pip(0) 
    &=\|u^\prime(\langle X , \pi^\opt \rangle)\|_{L^q(\P)}^{1-q} \cdot \left(\nabla_\pi^2 V(0) \right)^{-1} \cdot \frac{\pi^\opt}{\left| \pi^\opt \right|} \cdot \left( \E_{\P}\left[ \frac{\langle X,\pi^\opt \rangle\, u^{\prime\prime} (\langle X,\pi^\opt \rangle)+u^\prime(\langle X, \pi^\opt \rangle) }{\left| u^\prime(\langle X, \pi^\opt \rangle) \right|^{1-q} }\right] \right)
\end{align*}
was computed in Section \ref{sec:optim}.

We now conduct a first-order asymptotic analysis of the robust marginal utility price $\hatd(\delta)$. This will enable us to quantify a first-order premium paid for model-uncertainty. We derive the following result:

\begin{theorem}\label{thm:marg_sens}
\begin{enumerate}[(i)] Assume $g$ is continuously differentiable. 
\item Suppose $\pi^\opt=0$ and the following are satisfied:
\begin{enumerate}
\item if $p\in (1, \infty)$, then there exists $c>0$ such that 
\begin{align*}
|\nabla g(x)|\le c(1+|x|^{p-1})
\end{align*}
for all $x\in \mathcal{S}$,
\item if $p=\infty$, then there exists $\tilde{\delta}>0$ such that 
\begin{align*}
\E_\P\left[ \sup_{|z|\le \tilde{\delta}} |\nabla g(x+z)|\right] <\infty. 
\end{align*}
\end{enumerate}
Then the Davis price $\hatd(\delta)$ satisfies
\begin{align*}
\hatd ^\prime(0)&=-\left( \E_{\P} \left[ |\nabla g(x)|^q\right]\right)^{1/q}.
\end{align*}
\item Suppose $\pi^\opt\neq 0$ and Assumptions \ref{Ass:1}, \ref{Ass:2}, \ref{Ass:3} hold. Define
\begin{align*}
T(x):=\frac{\pi^\opt}{|\pi^\opt|} \frac{1}{|u^\prime (\langle x,\pi^\opt\rangle)|^{1-q}} \left(\E_{\P} \left[ |u^\prime (\langle X,\pi^\opt \rangle)|^q\right] \right)^{1/q-1}.
\end{align*}
Then
\begin{equation*}
\hatd^\prime(0) =\E_{\Q_u} \left[ R_u(\langle  X,\pi^\opt\rangle) \left( \langle T(X),\pi^\opt\rangle-\langle X,\pip(0)\rangle\right) \cdot \left(g(X)-\hatd\right)- \langle \nabla g (X) ,T(X)\rangle\right],
\end{equation*}
where $\Q_u$ was given in \eqref{eq:u mart mg} and $R_u(x)=-\frac{u''(x)}{u'(x)}$.
\end{enumerate}
\end{theorem}

\begin{proof}[Proof of Theorem \ref{thm:EUMsens}]
For the first claim, note that Lemma \ref{lem:degenerate} implies $\pi^\opt_\delta =0\in \Adm^{\opt,0}_{\delta}$ is the unique optimiser for all $\delta\ge 0$ and thus by equation \eqref{eq:marginal} in Theorem \ref{thm. marginal} the robust marginal utility price $\hatd(\delta)$ is given by
\begin{align*}
\hatd(\delta)=\inf_{\tilde{\P}\in B_{\delta}(\P)} \E_{\tilde{\P}}\left[  g(X)\right].
\end{align*}
Thus the claim follows from \cite[Theorem 2]{bartl2020robust}.\\
Assume now that $\pi^\opt \neq 0$ and fix the sequence $(\pi^\opt_\delta)_{\delta\ge 0}$ in $\Adm^{\opt, 0}_\delta$. By Lemma  \ref{lem:convergence}.\textit{(iii)} we have $\lim_{\delta\to 0} \pi_\delta^\opt =\pi^\opt$. Furthermore we know from Theorem \ref{thm. marginal} that
\begin{align*}
\hatd(\delta)=\frac{\E_{\P^\opt_\delta}\left[ u^\prime ( \langle X, \pi_\delta^\opt \rangle)\, g(X)\right]}{\E_{\P^\opt_\delta} \left[u^\prime (\langle X, \pi_\delta^\opt \rangle)\right]}
\end{align*}
for $\P_\delta^\opt \in B^{\opt,0}_\delta(\P)$ and all $\delta\ge 0$. 
Note that as
\begin{align*}
&\frac{\pi^\opt}{|\pi^\opt|} \frac{1}{|u^\prime (\langle x, \pi^\opt\rangle)|^{1-q}} \left(\E_{\P}\left[ |u^\prime (\langle X,\pi^\opt\rangle)|^q\right] \right)^{1/q-1}\\
&=\frac{\pi^\opt u^\prime (\langle x,\pi^\opt\rangle) }{(|\pi^\opt| u^\prime (\langle x,\pi^\opt\rangle))^{2-q}} |\pi^\opt|^{1-q} \left(\E_{\P}\left[ |u^\prime (\langle X, \pi^\opt\rangle)|^q\right]\right)^{1/q-1}
\end{align*}
we have
\begin{align*}
T(x)=\frac{\nabla_x f(x, \pi^\opt)}{|\nabla_x f(x, \pi^\opt)|^{2-q}} \left( \E_{\P}\left[| \nabla_x f(X, \pi^\opt)|^q\right] \right)^{1/q-1}
\end{align*}
for $f(x, \pi)=u(\langle x,\pi \rangle)$. By \cite[proof of Theorem 2]{bartl2020robust} we may write $\P^\opt_\delta=(x-\delta T(x)+o(\delta))_\#\P$. Applying the quotient rule to
\begin{align*}
\hatd(\delta)&=\frac{\E_{\P^\opt_\delta} \left[ u^{\prime}( \langle X, \pi^\opt_\delta \rangle)\, g(X)\right] }{ \E_{\P^\opt_\delta} \left[ u^{\prime}( \langle X, \pi^\opt_\delta \rangle)\right] }\\
&=\frac{\E_{\P}\left[ u^{\prime}( \langle X-\delta T(X)+o(\delta), \pi^\opt+\pip(0)\delta+o(\delta)\rangle)\, g(X-\delta T(X)+o(\delta))\right] }{\E_{\P}\left[ u^{\prime}( \langle X-\delta T(X)+o(\delta), \pi^\opt+\pip(0)\delta+o(\delta)\rangle)\right] }
\end{align*}
yields
\begin{align*}
\hatd^\prime(0)&=\frac{1}{\left(\E_{\P} \left[ u^{\prime}( \langle X, \pi^\opt \rangle)\right] \right)^2}\cdot \Bigg(\E_{\P}\left[  u^{\prime}( \langle X, \pi^\opt\rangle) \,g(X)\right]\cdot \E_{\P}\left[ u^{\prime\prime}( \langle X,\pi^\opt\rangle)\Big(\langle T(X),\pi^\opt \rangle -\langle X,\pip(0)\rangle\Big) \right]   \\
&\qquad  -  \E_{\P}\left[ u^{\prime}( \langle X,\pi^\opt\rangle)\right] \cdot \E_{\P}\Bigg[  u^{\prime \prime}(\langle X,\pi^\opt\rangle)\, g(X)\, \bigg[\langle T(X),\pi^\opt\rangle -\langle X,  \pip(0)\rangle\bigg]\\
&\qquad+u^\prime (\langle X,\pi^\opt \rangle) \,\langle \nabla g(X) ,T(X)\rangle\Bigg] \Bigg),
\end{align*}
which equals
\begin{align*}
&\frac{1}{\E_{\P}\left[ u^\prime (\langle X,\pi^\opt \rangle )\right] }\Bigg(\E_{\P}\Big[ u^{\prime \prime}(\langle  X,\pi^\opt\rangle) \Big(\langle T(X),\pi^\opt\rangle -\langle X,\pip(0)\rangle\Big) 
\cdot\left( \E_{\Q_u}\left[ g(X)\right]- g(X)\right)\,\Big] \Bigg)\\
&\qquad-\E_{\Q_u} \left[ \langle \nabla g (X) ,T(X)\rangle\right]
\end{align*}
and thus shows the claim. 
\end{proof}

\section{Data availability statement}
Data sharing not applicable to this article as no datasets were generated or analysed during the current study.

\appendix

\section{Proofs for Section \ref{sec:marginal_util}}\label{app:marg_util}

For completeness, we present here rigorous statements and proofs of the classical results on marginal utility pricing in absence of model uncertainty which we discussed in section \ref{sec:marginal_util}. We start with uniform integrability assumptions which yield the necessary control over $V(\epsilon,\d)$.
\begin{assumption}\label{Ass:5}
For each fixed $\d> 0$ the following holds:
\begin{enumerate}
\item For all $\pi\in \Adm$ there exists $k\in (0, \infty)$ such that the function $$x \mapsto u\left(-\epsilon+ \langle x-X_0, \pi \rangle +\frac{\epsilon}{\d}g(x) \right)$$ is dominated by a $\P$-integrable function uniformly for all $\epsilon\in [-k,k]$.
\item There exists a compact set $K$ such that $\pi^\opt\in \interior{K}$ and a constant $k\in (0,\infty)$ such that $\P$-a.s.  $$u^\prime\left(-\epsilon+ \langle x-X_0, \pi \rangle +\frac{\epsilon}{\d}g(x) \right)\le f(x)$$ for all $x\in \R^d$, $\pi\in K, \epsilon\in [-k, k]$, where $f:\R^d\to \R$ is chosen in such a way that $$x\mapsto f(x)\left(-1+\frac{g(x)}{\d}\right)$$ is a $\P$-integrable function.
\end{enumerate}
\end{assumption}

Assumptions \ref{Ass:4a} and \ref{Ass:5} imply the following:
\begin{lemma}\label{lem:strict_concave}
Let Assumptions \ref{Ass:4a} and \ref{Ass:5}.(1) hold.
For fixed $\epsilon\in\R, \d> 0$, the supremum in \eqref{eq:def_davis} is attained at a unique point, which we call $\pi^\opt(\epsilon)$. Furthermore $\pi^\opt(\epsilon )\to \pi^\opt(0)$ for $\epsilon \to 0$. 
\end{lemma}

\begin{proof}[Proof of Lemma \ref{lem:strict_concave}]
We first show existence of a maximiser. We follow  ideas outlined in \cite[Lemma 19]{bartl2020robust} (Note that we are only working in the case $\delta=0$ here): consider a maximising sequence $(\pi_n)_{n\in \N}$ for $V(\epsilon, \d)$ , i.e.
\begin{align*}
V(\epsilon, \d)=\lim_{n \to \infty} \E_\P\left[ u\left(-\epsilon+ \langle X, \pi_n \rangle +\frac{\epsilon}{\d}g(X) \right)\,\right].
\end{align*}
If $(\pi_n)_{n \in \N}$ is bounded, then after passing to a subsequence there is a limit, and the reverse Fatou lemma (recall that $u$ is bounded above) shows that this limit is a maximiser. It remains to argue why $\left(\pi_{n}\right)_{n \in \mathbb{N}}$ is bounded. Heading for a contradiction, assume that $\left|\pi_{n}\right| \rightarrow \infty$ as $n \rightarrow \infty$. After passing to a (not relabelled) subsequence, there exists $\tilde{\pi} \in \mathbb{R}^{d}$ with $|\tilde{\pi}|=1$ such that $\pi_{n} /\left|\pi_{n}\right| \to \tilde{\pi}$ as $n \to \infty$. As stated in Assumption \ref{Ass:4a} we have $$\P\left(\left\{X \in \mathcal{S}\ :\left\langle X , \tilde{\pi}\right\rangle<0\right\}\right)>0 .$$ As $u$ is bounded from above this shows that
$$
\lim_{n \to \infty} \E_{\P}\left[ u\left(\left\langle X, \pi_{n}\right\rangle+g(X) \right) \right] =- \infty,
$$
a contradiction. Uniqueness of optimisers follows again by Assumption \ref{Ass:4a} and strict concavity of $u$.\\
Now we prove the last claim. Heading for a contradiction, we assume that there exists a sequence $(\epsilon_n)_{n \in \N}$ converging to zero, such that $\pi^\opt(\epsilon_n)$ does not converge to $\pi^\opt.$ The exact same reasoning as above implies that $(\pi^\opt(\epsilon_n))_{n \in \N}$ is bounded, so that possibly after passing to a not relabelled subsequence there exists a limit $\tilde{\pi}\neq \pi^\opt$. The reverse Fatou lemma again implies that
\begin{align*}
V(0, \d)&> \E_{\P}\left[ u\left( \langle X, \tilde{\pi} \rangle \right)\right] \\
&\ge \limsup_{n \to \infty} \E_{\P}\left[ u\left(-\epsilon_n+ \langle X, \pi^\opt(\epsilon_n ) \rangle +\frac{\epsilon_n}{\d}g(X) \right)\,\right]=\limsup_{n\to \infty} V(\epsilon_n, \d).
\end{align*}
On the other hand, plugging in $\pi^\opt$ yields
\begin{align*}
\liminf_{n \to \infty}V(\epsilon_n, \d) &\ge \liminf_{n \to \infty} \E_{\P}\left[ u\left(-\epsilon_n+ \langle X, \pi^\opt \rangle +\frac{\epsilon_n}{\d}g(X) \right)\,\right] \\
&=\E_{\P}\left[ u\left( \langle X, \pi^\opt \rangle \right)\,\right]\\
&= V(0, \d),
\end{align*}
where we have used the dominated convergence theorem together with Assumption \ref{Ass:5}.(1). This yields a contradiction and concludes the proof.
\end{proof}

With the above  notational conventions, Mark Davis characterised the marginal utility price as follows:

\begin{theorem}[{cf. \cite[Theorem 3]{davis1997option}}]\label{Theorem:Davis}
Let Assumptions \ref{Ass:4a} and \ref{Ass:5} hold and assume
$$\E_{\P} \left[ u^\prime(\langle X-X_0, \pi^\opt\rangle)\right]>0.$$ 
Then $\hatd$ is unique and given by \eqref{eq:davis_main}.
\end{theorem}

\begin{proof}[Proof of Theorem \ref{Theorem:Davis}]
Recall that $x\mapsto u^\prime(x)$ is continuous as stated in the introduction. Assumption \ref{Ass:5} then enables the use of the dominated convergence theorem, so that
\begin{align*}
&\nabla_{\epsilon}\left( \E_{\P}\left[ u\left(-\epsilon+ \langle X, \pi \rangle +\frac{\epsilon}{\d}g(X) \right)\right]\right)\\
&=\E_{\P}\left[ u^\prime\left(-\epsilon+ \langle X, \pi \rangle +\frac{\epsilon}{\d}g(X) \right) \left(-1+\frac{g(X)}{\d}\right)\right].
\end{align*}
In particular the function 
\begin{align*}
(\epsilon, \pi) \mapsto &\nabla_{\epsilon}\left( \E_{\P}\left[ u\left(-\epsilon+ \langle X, \pi \rangle +\frac{\epsilon}{\d}g(X) \right)\right]\right)
\end{align*}
is uniformly continuous on $[-k,k]\times K$, as it is continuous and the set $[-k,k]\times K$ is compact. Furthermore we find that
\begin{align}\label{eq:abscont}
\begin{split}
&\left|\nabla_{\epsilon}\left( \E_{\P}\left[ u\left(-\epsilon+ \langle X, \pi \rangle +\frac{\epsilon}{\d}g(X) \right)\,\right]\right)\right|\\
&\le \E_{\P}\left[ \left| u^\prime\left(-\epsilon+ \langle X, \pi \rangle +\frac{\epsilon}{\d}g(X) \right)\right|\ \left|-1+\frac{g(X)}{\d}\right| \right]\\
&\le \E_{\P} \left[\left| f(X)\right|\ \left|-1+\frac{g(X)}{\d}\right| \right].
\end{split}
\end{align}
Thus the envelope theorem for arbitrary choice sets \cite[Theorem 2]{milgrom2002envelope} applies and so
\begin{align*}
\nabla_{\epsilon} V(0,\d)= \E_{\P}\left[ u^\prime\left(\langle X, \pi^\opt \rangle \right) \left(-1+\frac{g(X)}{\d}\right)\right].
\end{align*}
As the above formula holds for any $\d>0$, we conclude that
\begin{align*}
\hatd=\frac{ \E_{\P} \left[ u^\prime (\langle X,\pi^\opt\rangle) g(X)\right]}{\E_{\P}\left[ u^\prime(\langle X, \pi^\opt\rangle)\right]},
\end{align*}
which proves the result.
\end{proof}

\bibliographystyle{abbrvnat}
\bibliography{bib}
\end{document}